%% file: hw-revision1.tex
\documentclass{siamltex}
\usepackage{amsmath,amssymb,mathabx}
\usepackage{graphicx,esint,ulem}
\usepackage[usenames]{color}
\include{abbrevs}
\begin{document}

\title{Defect Modes and Homogenization of Periodic Schr\"{o}dinger Operators}

\author{M. A. Hoefer\footnotemark[2] \and
  M. I. Weinstein\footnotemark[3]}

\maketitle

\renewcommand{\thefootnote}{\fnsymbol{footnote}}

\footnotetext[2]{Department of Mathematics, North
  Carolina State University, Raleigh, NC 27695; mahoefer@ncsu.edu}
\footnotetext[3]{Department of Applied Physics and Applied
  Mathematics, Columbia University, New York, NY 10027; miw2103@columbia.edu}

\renewcommand{\thefootnote}{\arabic{footnote}}

\begin{abstract}
  We consider the discrete eigenvalues of the operator
  $H_\eps=-\Delta+V(\x)+\eps^2Q(\eps\x)$, where $V(\x)$ is periodic
  and $Q(\y)$ is  localized on $\R^d,\ \
  d\ge1$. For $\eps>0$ and sufficiently small, discrete eigenvalues
  may bifurcate (emerge) from spectral band edges of the periodic Schr\"odinger
  operator, $H_0 = -\Delta_\x+V(\x)$, into spectral gaps.  The nature
  of the bifurcation depends on the homogenized Schr\"odinger operator
  $L_{A,Q}=-\nabla_\y\cdot A \nabla_\y +\ Q(\y)$. Here, $A$ denotes
  the inverse effective mass matrix, associated with the spectral band
  edge, which is the site of the bifurcation.
\end{abstract}



\pagestyle{myheadings} \thispagestyle{plain} \markboth{M. A. Hoefer
  and M. I. Weinstein}{Defect Modes and Homogenization of Periodic
  Schr\"{o}dinger Operators}

\section{Introduction and Outline}
\label{sec:intr-outl-results}

Self-adjoint elliptic partial differential operators with periodic
coefficients {\it e.g.}  the Schr\"odinger
operator with a periodic potential, the time-harmonic Helmholtz
equation with variable refractive index, and the time-harmonic Maxwell
equations with variable dielectric and permeability tensors, play a
central role in wave propagation problems in classical and quantum
physics. The spectrum of such operators, characterized by
Floquet-Bloch theory \cite{RS4,Kuchment-01,Eastham-73}, consists of
the union of closed intervals (spectral bands). The eigenstates are
{\it extended} (not localized) and form a complete set with respect to
which any function in $L^2(\R^d)$ may be represented.
 
In many problems in fundamental and applied physics, periodic media
are perturbed by spatially localized defects. These may appear as
random imperfections in a media, {\it e.g.} a defect in a crystal, or
in engineering applications, they may be introduced deliberately in
order to influence wave propagation
\cite{Ashcroft-Mermin:76,Joannopoulos-etal:08}.  Since the essential
spectrum is unchanged by a sufficiently localized and smooth
perturbation (Weyl's theorem, \cite{RS4}), typical localized
perturbations will only introduce eigenvalues in {\it spectral gaps}
of the spectrum with associated localized {\it defect modes}.

This paper is concerned with a class of localized (defect)
perturbations to a periodic Schr\"odinger operator of the form:
\begin{equation}
  \label{eq:81}
  H_\eps\ =\  -\Delta_\x + V(\x) + \eps^2Q(\eps \x),
  \nn\end{equation}
where $V(\x)$ is periodic on $\R^d$, $Q(\y)$ decays as $|\y|$
tends to infinity and $\eps$ is a small parameter.

Our main result, Theorem \ref{sec:main-theorem}, concerns the
perturbed eigenvalue problem
\begin{equation}
  \label{eq:8}
  \begin{split}
    H_\eps u_\eps = \mu_\eps u_\eps, \quad u_\eps \in H^1(\real^d) ,
  \end{split}
\end{equation}
for $\varepsilon$ positive and sufficiently small.
See section \ref{sec:main-results} for hypotheses on the periodic
potential, $V$, and the localized perturbation, $Q$.
 
For $\eps$ sufficiently small, we prove the bifurcation of discrete
eigenvalues into the spectral gaps, associated with the unperturbed
operator, $H_0=-\Delta + V(\x)$. For any given spectral band edge, we
give detailed expansions with error estimates for the perturbed
eigenvalues and corresponding localized eigenfunctions in terms of the
eigenstates of a {\it homogenized} Schr\"odinger operator
\begin{equation}
  L_{A,Q}\ =\ -\sum_{j,l=1}^d \ \frac{\partial}{\partial y_j}\ A_{jl}\
  \frac{\partial}{\partial y_l}\ +\ Q(\y). 
\label{LAQ}
\end{equation}
Here, $A_{jl}$ denotes the {\it inverse effective mass matrix},
associated with the particular band edge from which the bifurcation
occurs; see Theorem \ref{sec:main-theorem}. $A_{jl}$, derivable by
formal multiple scale expansion (see section
\ref{sec:homog-multi-scale}), is expressible in terms of the band edge (Floquet-Bloch) 
eigenstate.  It is  proportional to the Hessian
matrix $D^2E_{b_*}(\kv_*)$ of the band dispersion function, associated with $-\Delta+V(\x)$,  evaluated
at the band edge $E_*=E_{b_*}(\kv_*)$.  \bigskip
 
Referring to the schematics of figures \ref{fig:band_edge} and
\ref{fig:multiple_degenerate}, we discuss our results.
\begin{itemize}
\item Suppose the inverse effective mass matrix, $A$, is positive definite and
  assume $L_{A,Q}$ has an eigenvalue, $e_{A,Q}<0$. This occurs if
  $Q(\y)$ is a 
  ``down-defect'' (sufficiently ``deep'' in dimensions $d\ge3$) as in figure \ref{fig:band_edge}.a. In this case,
  Theorem \ref{sec:main-theorem} asserts the existence of an
  eigenvalue at $E_*+\eps^2e_{A,Q}+\cO(\eps^3)<E_*$.
\item Now suppose the inverse effective mass matrix, $A$, is 
  negative definite and $L_{A,Q}$ has an eigenvalue,
  $e_{A,Q}>0$. This occurs if $Q(\y)$ is a  ``up-defect'' (sufficiently ``high'' in
  dimensions $d\ge3$) as in figure
  \ref{fig:band_edge}.b. In this case, Theorem \ref{sec:main-theorem}
  asserts the existence of an eigenvalue at
  $E_*+\eps^2e_{A,Q}+\cO(\eps^3)>E_*$
\item Fig.~\ref{fig:multiple_degenerate} shows a more general band
  edge bifurcation when $L_{A,Q}$ has three eigenvalues $e_{A,Q}^{(1)}
  < e_{A,Q}^{(2)} < e_{A,Q}^{(3)}$, the largest of which is degenerate
  with multiplicity three.  Theorem \ref{sec:main-theorem} asserts the
  existence of five ordered eigenvalues at $E_* + \eps^2 e_{A,Q}^{(j)}
  + \cO(\eps^3),\ j=1,2$ and $E_* + \eps^2 e_{A,Q}^{(3)} + \eps^3
  \mu_3^{(k)} + \cO(\eps^4),\ k=1,2,3$.
\end{itemize}
 \begin{figure}
  \centering
  \includegraphics[width=12.5cm,height=6.25cm]{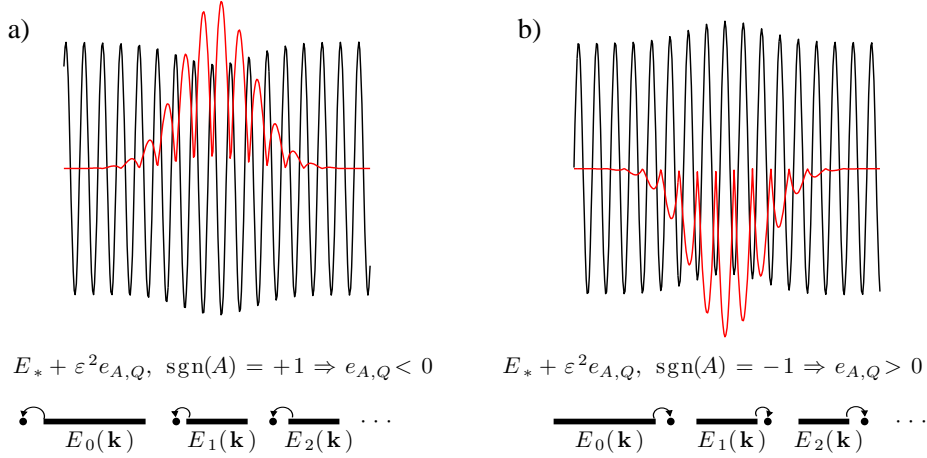}
  \caption{a) Periodic structure with ``down defect'' and
    corresponding localized eigenstate for the case of positive
    definite effective mass tensor.  b) Periodic structure with ``up
    defect'' and corresponding localized eigenstate for the case of
    negative definite effective mass tensor.  Below are shown
    eigenvalue bifurcations from band edges of the form
    $E_{b_*}(\kv_*) = E_*$.}
  \label{fig:band_edge}
\end{figure}

 \begin{figure}
  \centering
  \includegraphics{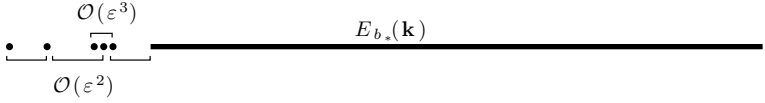}
  \caption{Schematic of band edge bifurcations in the case where the inverse effective mass matrix,  $A$, is positive
    definite. The homogenized operator, $L_{A,Q}$, is assumed to have two simple eigenvalues and one degenerate
    eigenvalue with multiplicity three.}
  \label{fig:multiple_degenerate}
\end{figure}
\bigskip

\subsection{Outline of the paper and overview of the proof}
\label{sec:outline}
Section \ref{sec:spectr-theory-peri} summarizes the required spectral
theory for Schr\"odinger operators with periodic potentials and
introduces variants of the classical Sobolev space, $H^s(\R^d)$, which provide a
natural functional analytic setting. Section \ref{sec:main-results}
contains the hypotheses on $V$ and $Q$ and the statement of our main
theorem, Theorem \ref{sec:main-theorem}. In section
\ref{sec:homog-multi-scale} we present a {\it formal} multiple scale /
homogenization expansion in which we systematically construct
bifurcating eigenstates and eigenvalues to any prescribed order. In
section \ref{sec:proof-theorem-} we prove Theorem
\ref{sec:main-theorem}. In particular, we study the equations
governing the correction, $\Psi^\eps$ to the $N-$ term multiple scale
expansion.

To obtain error bounds of suitably high order in $\eps$, we use a
Lyapunov-Schmidt approach.  Specifically, we decompose the error into
Floquet-Bloch modes associated with energies lying near the spectral
band edge, $E_*$, and those lying ``far" from $E_*$:
$\Psi^\eps=\Psi^\eps_{\rm near}+\Psi^\eps_{\rm far}$.  
$\Psi^\eps_{\rm near}$ has the character of a wave-packet, spectrally
supported on a small interval with endpoint $E_*$. The next step is to
solve for $\Psi^\eps_{\rm far}$ as a functional of the ``parameter''
$\Psi^\eps_{\rm near}$, with appropriate bounds. Substitution of
$\Psi^\eps_{\rm far}[\Psi^\eps_{\rm near}]$ into the near equation
implies a closed equation for $\Psi^\eps_{\rm near}$.  With strong
motivation from the structure of terms in the multiple scale
expansion, we appropriately rescale, solve via the implicit function
theorem, and estimate $\Psi^\eps_{\rm near}$. The approach we take has
been applied in the context of the nonlinear Schr\"odinger /
Pitaevskii equation in
\cite{sparber_effective_2006,Pel-Sch:07,dohnal_coupled_2008,dohnal_coupled-mode_2009,Ilan-Weinstein:10}.\bigskip

{\bf Previous work for linear Schr\"odinger operators:} Bifurcation of
eigenvalues from the edge of the continuous spectrum for Schr\"odinger
operators with small decaying potentials, corresponding to weak
defects in dimensions one and two for the case of a homogeneous medium
or vacuum ($V\equiv0$), was studied in \cite{Simon:76}.  Conditions
ensuring the existence of eigenvalues in the gaps of periodic
potentials were obtained in \cite{Alama-Deift-Hempel:89} and
\cite{Figotin-Klein:97,Figotin-Klein:98}, using the Birman-Schwinger
(integral equation) formulation of the eigenvalue problem.
  Homogenization theory was applied to
obtain eigenvalues in the spectral gaps of a class of periodic
divergence form elliptic operators, governing localized states in high
contrast media in \cite{Kamotski-Smyshlyaev:06,Cherdantsev:09}.  An
elementary variational argument in spatial dimensions one and two,
yielding general conditions for the existence of discrete modes in
spectral gaps of periodic potentials, was recently presented in
\cite{Johnson-etal:10}.  More general, variational methods can be
applied to obtain defect modes which are obtained as infinite
dimensional saddle points of {\it strongly indefinite} functionals;
see, for example, \cite{dolbeault_eigenvalues_2000}.

{\it Our results concern a particular class of {\it weak defects},
  slowly varying and of small amplitude: $\eps^2Q(\eps\x)$, which give
  rise to defect modes in \underline{any} spatial dimension.  We note
  that the one- and two-term truncated multi-scale homogenization
  expansion of defect modes, which we construct, are natural trial
  functions for a variational proof of existence of ground states; see
  the discussion in Appendix \ref{sec:vari-exist-proof}.  Note also
  that the scaling of the perturbing potential, $\eps^2 Q(\eps\x)$,
  also arises naturally in solitary standing wave (``soliton defect
  mode'') bifurcations from band edges of periodic potentials in the
  nonlinear Schr\"odinger / Gross-Pitaevskii equation
  \cite{Ilan-Weinstein:10}.}
  
Homogenization theory has been used to study periodic elliptic
divergence form operators near spectral band edges in
\cite{Birman:04,Birman-Suslina:06,allaire_periodic_2008}.
Homogenization results for the {\it time-dependent} Schr\"{o}dinger
equation with a scaling, equivalent to the one considered here, were
obtained by two-scale convergence methods in
\cite{allaire_homogenization_2005}; see also
\cite{poupaud_semi-classical_1996,bechouche_semiclassical_1999,allaire_periodic_2008}.
In \cite{allaire_homogenization_2005} the contrast between the scaling
we use and the semi-classical scaling is discussed.  These results
establish the validity of the homogenized time-dependent Schr\"odinger
equation on certain {\it finite} time scales. The results of the
present paper focus on a subclass of solutions, bound states, which
are controlled on {\it infinite} time scales.

Finally, we mention work on effective classical electron motion in
solid state physics, derived from the Schr\"odinger equation for an
electron in a spatially periodic Hamiltonian, perturbed by spatially
slowly varying electrostatic and magnetic potentials
\cite{Nenciu:91,PanatiSpohnTeufel:02,PanatiSpohnTeufel:03}, in a
semi-classical limit.  \bigskip

\textbf{Acknowledgments:} This research was initiated while MAH was an
NSF Postdoctoral Fellow under DMS-08-03074 in the Department of
Applied Physics and Applied Mathematics at Columbia University. MIW
was supported in part by NSF grant DMS-07-07850 and DMS-10-08855. MIW
would also like to acknowledge the hospitality of the Courant
Institute of Mathematical Sciences, where he was on sabbatical during
the preparation of this article.

\subsection{Notation and conventions}
\label{sec:notation}

We note that we may, without loss of generality, restrict to the case
where the fundamental period cell is $\Omega=[0,1]^d$. Indeed, let
${\cal B}$ denote the fundamental period cell, spanned by the linearly
independent vectors $\{\br_1, \dots, \br_d\}$ and define the constant
matrix $\cR^{-1}$ to be the matrix whose $j^{\rm th}$ column is
$\br_j$. Then, under the change of coordinates
$\bx\mapsto\bz=\cR\bx$,
\begin{align}
  &-\nabla_\bx\cdot \nabla_\bx\ +\ V(\bx)\ {\rm acting\ on}\
  L_\textrm{per}^2(\cell)\ \ {\rm transforms\ to}\nn\\ 
  & -\nabla_\bz\cdot \alpha\ \nabla_\bz
  + \tilde{V}(\bz) \equiv\ -\ \sum_{i,j=1}^d \alpha_{ij}\
  \frac{\D^2}{\D z_i\D z_j}\ +\ \tilde{V}(\bz)\nn\\ 
  & {\rm acting\ on}\ L^2_\textrm{per}\left([0,1]^d\right)\
  {\rm where}\nn\\
  &\alpha= \cR \cR^T,\ \
  \tilde{V}(\bz)=V\left(\cR^{-1}\bz\right),\ \ \bx=\cR^{-1}\bz\ .
  \nn
\end{align}

\bigskip

\begin{enumerate}
\item Integrals with unspecified region of integration are assumed to be taken over $\R^d$,
 {\it i.e.} $\int f\ =\int_{\R^d} f(\x) d\x$.
\item For $f, g\in L^2$, the Fourier transform and its inverse are given by:
\begin{align}\label{FTdef}
  \mathcal{F}\{f\}(\kv) &\equiv \widehat{f}(\kv) = \int e^{-2\pi
    i \kv \cdot \x} f(\x)\, d\x,\\ 
  \mathcal{F}^{-1}\{g\}(\x)  &\equiv\ \widecheck{g}(\kv) = \int e^{2\pi
    i \x \cdot \kv} g(\kv)\, d\kv .\nn
\end{align}
Thus, $ \mathcal{F}\ \mathcal{F}^{-1}=Id$.  \item $\Omega =
  [0,1]^d$ is the fundamental period cell, $\Omega^* = [-1/2,1/2]^d$
  is the dual fundamental cell or Brillouin zone
\item $1_A(\x)$ is the indicator function of the set $A$; $\chi(|\kv|
  \le a) \equiv 1_{\{\kv \in \Omega^* \, : \, |\kv | \le a
    \}}$
\item The repeated index summation convention is used throughout
\item Fourier spectral cutoff:
\begin{align}
  \chi(|\nabla | < a) G(\x) &\equiv\ \left(\cF^{-1} \chi(|\kv|<a)
    \cF\right) G \ =\ \int e^{2\pi i \x \cdot \kv}
  \chi(|\kv| < a) \widehat{G}(\kv) d\kv\nn\\
  \end{align}
\item $\cT$ and $\cT^{-1}$ denote the Gelfand-Bloch transform and its
  inverse; see section~\ref{sec:spectr-theory-peri}.
\item Bloch spectral cutoff:
  \begin{align}
    \chi(|\nabla | < a) G(\x)&\equiv \mathcal{T}^{-1} \left \{ \sum_{b \ge
        0} \chi(|\cdot | 
      < a\delta_{bb_*}) \blochn{b}{G}(\cdot) \,p_b(\x; \cdot) \right
    \}(\x),\nn
  \end{align}
  where $b_*$ is the index of the spectral band under
    consideration,
\item $H^s=H^s(\R^d)$ is the Sobolev space of order $s$
  \begin{equation}
    \label{eq:86}
    \| f \|_{H^s}^2 \equiv \sum_{|\alpha| \le s} \| \partial^\alpha f
    \|^2_{L^2} \sim \| \widehat{f} \|_{L^{2,s}}^2 = \int_{\real^d} (1
    + |\kv|^{s\over 2})^2 | \hat{f}(\kv) |^2 \, d\kv 
  \end{equation}
\end{enumerate}


\section{Spectral Theory for Periodic Potentials}
\label{sec:spectr-theory-peri}


In this section we summarize basic results on the spectral theory of
Schr\"odinger operators with periodic potentials; see, for example,
\cite{RS4,Kuchment-01,Eastham-73}.
\medskip

\noindent{\bf Gelfand-Bloch transform:}\ \ Given $f\in L^2(\real^d)$,
we introduce the transform $\mathcal{T}$ and its inverse as follows
\begin{align}
  \label{eq:43}
  \mathcal{T}\{f(\cdot)\}(\x;\kv) = &\tilde{f}(\x;\kv) = \sum_{\z \in
    \integer^d} e^{2\pi i \z \cdot \x} \widehat{f}(\kv+\z),  \\
  \mathcal{T}^{-1}\{\tilde{f}(\x;\cdot) \}(\x) &= 
  \int_{\Omega^*} e^{2\pi i \x \cdot \kv} \tilde{f}(\x;\kv) d\kv .
\end{align}
One can check that $ \mathcal{T}^{-1} \mathcal{T}=Id$.  

Two important properties of the transformation $ \mathcal{T}$ are $
\mathcal{T}\D_{x_j}f=\left(\D_{x_j}+2\pi i k_j\right) \mathcal{T}f$
and $ \left(\mathcal{T} e^{2\pi i \kv\cdot }f\right)(\x,\kv)=e^{2\pi i
  \kv\cdot \x }\mathcal{T}f(\x,\kv)$. It follows that
\begin{align}
  \left( \mathcal{T} \Phi(\nabla) e^{2\pi i \kv\cdot }f\right)(\x,\kv)
  &= e^{2\pi i \kv\cdot x} \Phi(\nabla+2\pi i\kv)
  \left(\mathcal{T}f\right)(\x,\kv)\label{TF-conj}\\
  \label{eq:116}
\mathcal{T} \left(v(\cdot) f\right)(\x,\kv) &= v(\x)
  \left(\mathcal{T} f\right)(\x,\kv),\ \ {\rm if\ } v\ {\rm is\
    periodic.}
\end{align}
\medskip

\noindent {\bf Floquet-Bloch states:} We seek solutions of the eigenvalue
equation
\begin{equation}
  \left(-\Delta + V(\x)\right) u(\x) = E u(\x)\label{ev-eqn}
\end{equation}
in the form $u(\x;\kv)=e^{2\pi i\kv\cdot\x}p(\x;\kv),\ \kv\in\Omega^*$
where $p(\x;\kv)$ is periodic in $\x$ with fundamental period cell
$\Omega$. $p(\x;\kv)$ then satisfies the periodic elliptic boundary
value problem:
\begin{equation}
  \left(-\left(\nabla+2\pi i\kv\right)^2 + V(\x)\right) p(\x;\kv) = E(\kv)
  p(\x;\kv)\label{k-ev-eqn}, \quad \x \in \T^d .
\end{equation}
For each $\kv\in\Omega^*$, the eigenvalue problem \eqref{k-ev-eqn} has
a discrete set of eigenpairs \\
$\{\ p_b(\x;\kv),\ E_b(\kv)\ \}_{b\ge0}$
which form a complete orthonormal set in $L^2_{\rm per}(\Omega)$.  The
spectrum of $-\Delta + V(\x)$ in $L^2(\real^d)$ is the union of closed
intervals
\begin{equation}
  \label{eq:87}
  \textrm{spec}(-\Delta + V) = \bigcup_{b
        \ge 0, ~ \kv \in \Omega^*} E_b(\kv) .
\end{equation}
We will study the bifurcation of eigenvalues from the band edge
\begin{equation}
  \label{eq:110}
  E_* \equiv E_{b_*}(\kv_*), \quad k_{*,j} \in \{0, 1/2 \}, \quad j =
  1, \ldots, d ,  
\end{equation}
with the associated, real-valued band edge eigenfunction
\begin{equation}
  \label{eq:52}
  w(\x) \equiv e^{2\pi i \kv_* \cdot \x} p_{b_*}(\x;\kv_*) \in L^2(\Omega) .
\end{equation}
For example, the lowest band edge is $E_0(0)$ and the associated
eigenfunction is periodic $p_0(\x + \mathbf{e}_j;0) = p_0(\x;0)$,
$j=1,\ldots, d$ for the standard Cartesian basis vectors
$\{\mathbf{e}_j\}_{j=1}^d$.
\begin{rem}
  \label{sec:spectr-theory-peri-1}
  Since 
  \begin{equation}
  w(\x + \mathbf{e}_j) = e^{2\pi i k_{*,j}} w(\x) = s_j w(\x),
  \label{w-symm}
  \end{equation}
  where $s_j = +1$ if $k_{*,j} = 0$ and $s_j = -1$ otherwise, the
  natural function space to work in is $\Ls$, i.e.~$f \in \Ls$ if $f
  \in L^2(\Omega)$ and $f(\x + \mathbf{e}_j) = s_j f(\x)$.  Without
  loss of generality, and for ease of presentation, we focus on the
  case where $s_j = +1$, $j= 1, \ldots, d$ so that $\Ls =
  L^2_\textrm{per}(\Omega)$, the space of square integrable, periodic
  functions.  This implies that $\kv_* = 0$.  The more general case in
  eq.~(\ref{eq:110}) can be handled by taking $\kv \to (\kv - \kv_*)$
  and interpreting values of $\kv$ reflected about the boundary of
  $\Omega^*$.  The simplicity of $E_*$ and the relation $\nabla
  E_{b_*}(\kv_*) = 0$ (see Hypothesis H2 in
  Sec.~\ref{sec:main-results}) implies that $E(\kv' + \kv_*)$ can be
  extended as an even function of $k'_j = k_j - k_{j,*}$ for $j = 1,
  \ldots, d$ \cite{RS4}.
\end{rem}

We will make repeated use of the following self-adjoint operator
\begin{equation}
  \label{eq:54}
  L_* \equiv - \Delta + V(\x) - E_*: H_\textrm{per}^2(\Omega) \to
  L^2(\Omega) .
\end{equation}
\medskip

\noindent{\bf Projections $\mathcal{T}_b$ and Completeness of Floquet
  Bloch states:}\ Define
\begin{equation}
  \label{eq:82}
  \mathcal{T}_b\{f\}(\kv) \equiv \langle
  p_b(\cdot;\kv), \tilde{f}(\cdot;\kv) \rangle_{L^2(\Omega)} \equiv \int_{\Omega}
  \overline{p_b(\x;\kv)}\  \tilde{f}(\x;\kv)\, d\x .
\end{equation}
By completeness of the $\{p_b(\x;\kv)\}_{b\ge0}$
\begin{equation*}
  \tilde{f}(\x;\kv) = \sum_{b\ge0}  \mathcal{T}_b\{f\}(\kv)\ p_b(\x;\kv)
\end{equation*}
Furthermore, applying ${\cal T}^{-1}$  we have 
\begin{align}
  f(\x) &= \sum_{b\ge0}  \int_{\Omega^*}\ \mathcal{T}_b\{f\}(\kv)\
  u_b(\x;\kv)\ d\kv 
  \label{complete1}\\ &= \sum_{b\ge0}\ \int_{\Omega^*}\ \langle
  u_b(\cdot;\kv),f\rangle_{L^2(\real^d)}\ u_b(\x;\kv)\ d\kv ,
  \label{complete2}
\end{align}
where $u_b(\x;\kv) = e^{2\pi i \kv \cdot \x} p_b(\x;\kv)$.  The second
equality follows from an application of the Poisson summation formula.
\medskip

\noindent{\bf Sobolev spaces and the Gelfand-Bloch transform:}\\
Recall the Sobolev space, $H^s$, the space of functions with
square-integrable derivatives of order $\le s$. Since $E_0(0)=\inf\
{\rm spec}(-\Delta+V)$, then $L_0=-\Delta+V(\x)-E_0(0)$ is a
non-negative operator and $H^s(\R^d)$ has the equivalent norm defined
by
\begin{equation}
  \|\phi\|_{H^s}\ \sim \ \| (I+L_0)^{s\over2}\phi\|_{L^2}\
  \nn\end{equation}
Introduce the space (see, {\it e.g.}
\cite{Pel-Sch:07,dohnal_coupled-mode_2009}) 
\begin{equation}
  \label{eq:67}
  {\cal X}^s = L^2(\Omega^*,l^{2,s}),
\end{equation}
with norm
\begin{equation}
  \label{eq:68}
  \| \tilde{\phi} \|^2_{{\cal X}^s} \equiv \int_{\Omega^*}
  \sum_{b = 0}^\infty \left(1 + |b|^{2\over d}\right)^s  |
  \mathcal{T}_b \{ \phi \} (\kv) |^2 \, d\kv . 
\end{equation}
Now note that 
\begin{align}
  \|\phi\|_{H^s}^2 &\sim\| (I+L_0)^{s\over2}\phi\|_{L^2}^2\ \nn\\
  &=\ \left\|\
    \int_{\Omega^*}\ e^{2\pi i\kv\cdot }\sum_{b\ge0}
    \mathcal{T}_b\{\phi \}(\kv)
    \left(1+ E_b(\kv) - E_0(0) \right)^{\frac{s}{2}} \ p_b(\cdot,\kv)\
  d\kv \right\|_{L^2}^2\nn\\ 
  &=\sum_{b\ge0}\ \int_{\Omega^*}\ |\mathcal{T}_b \{\phi\}(\kv)|^2\
  |1+ E_b(\kv) - E_0(0) |^{s}\ d\kv
  \nn\\
  &\sim\ \sum_{b\ge0}\ \left(1+|b|^{2\over d}\right)^{s}\
  \int_{\Omega^*}\ |\mathcal{T}_b \{\phi\}(\kv)|^2\ d\kv
  \nn\\
  &\equiv \| \tilde{\phi} \|^2_{{\cal X}^s} .
  \label{isonorm-pf}
\end{align}
The second to last line follows from the Weyl asymptotics
$E_b(\kv)\sim b^{2\over d}$ \cite{Hor4}.  Thus we have
\begin{proposition}
  $H^s(\real^d)$ is isomorphic to ${\cal X}^s$ for $s \ge
  0$. Moreover, there exist positive constants $C_1$, $C_2$ such that
  for all $\phi\in H^s(\real^d)$
  \begin{equation}
    \label{eq:69}
    C_1 \|\phi\|_{H^s(\real^d)} \leq  \|
    \tilde{\phi} \|_{{\cal X}^s}  \leq C_2 \|\phi\|_{H^s(\real^d)} \ . 
  \end{equation}
\end{proposition}

\section{Main Results}
\label{sec:main-results}
In this section we give a precise formulation of our main theorem,
Theorem \ref{sec:main-theorem}.  The following are our assumptions.
\begin{itemize}
\item[H1] \textbf{Regularity.}  $V \in
  L^\infty_{\textrm{per}}(\Omega)$, $Q \in H^\sigma(\real^d)$ for
  $\sigma > d$, $E_{b_*} \in C^3(\Omega^*)$, and $p_{b_*} \in
  L^2_\textrm{per}(\Omega; C^3(\Omega^*))$.
\item[H2] \textbf{Band edge.}  $E_* \equiv E_{b_*}(\kv_*)$,
  where $\kv_*$ is an endpoint of the  $b_*^{th}$ band 
  such that \subitem (a)\ $E_*$ is a simple eigenvalue with
  corresponding eigenfunction 
  \begin{equation}
  w(\x) \equiv e^{2\pi i \kv_* \cdot \x}
  p_{b_*}(\x;\kv_*) \in H^2_\textrm{per}(\Omega)
  \nn\end{equation}
   and normalization
  $\| w \|_{L^2(\Omega)} = 1$.  \subitem (b)\ $\nabla E_{b_*}(\kv_*) =
  0$. \subitem (c)\ The Hessian matrix,
 \begin{equation}
 A\equiv \frac{1}{8\pi^2} D^2E_{b_*}(\kv_*),
 \label{hessian}
 \end{equation}
 is sign definite.
\item[H3] \textbf{Existence of eigenvalue to homogenized equation.}\\
  Introduce the homogenized operator
  \begin{equation}
    L_{A,Q}\ \equiv  -\nabla_\y\cdot A\ \nabla_\y + Q(\y)\ =\
    -\sum_{j,l}\frac{\D}{\D y_j}A_{jl} \frac{\D}{\D y_l} + Q(\y) 
    \label{LAQ-def}
  \end{equation}
  Set $\textrm{sgn}(A)=+1$ if $A$ is positive definite and
  $\textrm{sgn}(A)=-1$ if $A$ is negative definite.  Assume $ L_{A,Q}$
  has a simple eigenvalue $e_{A,Q}$ with $\textrm{sgn}(A) e_{A,Q} < 0$
  and corresponding eigenfunction $F_{A,Q}(\y)\in H^2(\real^d)$; i.e.
 \begin{equation}
   \label{eq:111}
   L_{A,Q} F_{A,Q} = e_{A,Q} F_{A,Q}, \quad
   \int_{\real^d} F_{A,Q}^2(\y) d\y = 1, \quad \textrm{sgn}(A)\, e_{A,Q} < 0;
 \end{equation}
 see figure \ref{fig:q_spectrum}(a).
 \begin{figure}
   \centering
   \includegraphics{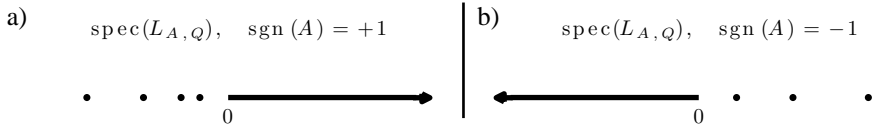}
   \caption{Discrete and continuous spectrum of $L_{A,Q}$.  a)
     Positive definite effective mass tensor.  b) Negative definite
     effective mass tensor.}
   \label{fig:q_spectrum}
 \end{figure}
\end{itemize}
\begin{rem}
  For further details regarding the smoothness properties of $E_{b_*}$
  and $p_{b_*}$ with respect to $\kv$, we refer the reader to
  \cite{RS4,wilcox_theory_1978}.  It can be verified that
  hypothesis H2 holds in one dimension at all band edges
  \cite{Eastham-73} and at the lowest band edge in arbitrary
  dimensions \cite{Kirsch-Simon:87}.  Band edges with multiplicity
  greater than one exist, e.g.~for the separable potential $V(\x) =
  \sum_{j=1}^d V_1(x_j)$, $d \ge 2$.
\end{rem}

\medskip

\begin{theorem} 
  \label{sec:main-theorem}
  \nit (1) {\bf Positive definite effective mass tensor:}\ Assume
  hypotheses H1-H3, with $\textrm{sgn}(A)=+1$. Then, there exists
  $\eps_0 > 0$ such that for all $0<\eps<\eps_0$, \eqref{eq:8} has an
  eigenpair $\mu_\eps,\ u_\eps(\x;\mu_\eps) \in
  H^2(\real^d)$. $\mu_\eps$ lies in the spectral gap of
  $-\Delta+V(\x)$ at a distance $\cO(\eps^2)$ below the spectral band
  edge having  $E_*$ as its left endpoint.
  
  Moreover, to any order in $\eps$, this solution can be approximated
  by the two-scale homogenization expansion, see Eq.~\eqref{eq:46}, \eqref{eq:42},
  with error estimate
  \begin{align}
    \label{eq:96}
    \Big \| u_\eps(\cdot;\mu_\eps) - \sum_{n=0}^N \eps^n U_n(\cdot,\eps \cdot)
    \Big \|_{H^2(\real^d)} &\leq \eps^{N-1} C,\nn\\
    \Big | \mu_\eps - E_* - \eps^2 e_{A,Q} - \sum_{n=3}^N \eps^n \mu_n \Big |
    &\leq \eps^{N+1} C,
  \end{align}
  for all $N \geq 4$ and some constant $C>0$, which is independent of $\eps$.
  
  \nit (2) {\bf Negative definite effective mass tensor:}\ Assume
  hypotheses H1-H3, with $\textrm{sgn}(A)=-1$. Then, the statement of
  part (1) applies, but now $\mu_\eps$ lies in the spectral gap of
  $-\Delta+V(\x)$ at a distance $\cO(\eps^2)$ above the spectral band
  edge having $E_*$ as its right endpoint.
\end{theorem}
\medskip

Theorem \ref{sec:main-theorem} extends to the case where $L_{A,Q}$ has
multiple and/or degenerate eigenvalues with bifurcations from band
edges with $\kv_* \ne 0$, as discussed in the following two remarks.

\begin{rem} {\bf General band edge bifurcations:} As discussed in
  Remark \ref{sec:spectr-theory-peri-1}, Theorem
  \ref{sec:main-theorem} generalizes to band edges where $\kv_* \ne 0$
  satisfying eq.~(\ref{eq:110}) so that $w(\x) \in \Ls$.
\end{rem}

\begin{rem} {\bf Multiple simple eigenvalues:} Note that if $L_{A,Q}$
  has $M$ (finitely many) eigenvalues, $e_{A,Q}^{(j)},\ j=1,\dots,M$
  of multiplicity one, then Theorem \ref{sec:main-theorem} applies
  directly. Specifically, there exists $\tilde\eps_0>0$ such that for
  all $0 < \eps<\tilde\eps_0$, there are eigenvalue / eigenvector
  branches
  $\eps\to\left(\mu_\eps^{(j)},u_\eps(\cdot;\mu_\eps^{(j)})\right)$.
  This behavior is shown in Fig.~\ref{fig:multiple_degenerate} with
  two simple eigenvalue branches with spacings $\cO(\eps^2)$.
\end{rem}

\begin{rem} {\bf Branches emanating from degenerate eigenvalues of
    $L_{A,Q}$:}
  In spatial dimensions, $d > 1$, the operator $L_{A,Q}$ may have
  degenerate eigenvalues, {\it e.g.}~if there is symmetry in $Q(\y)$.
  Suppose $e_{A,Q}$ has multiplicity $M$. Then, since $L_{A,Q}$ is
  self-adjoint, $e_{A,Q}$ perturbs, generically, to $M$ distinct
  branches. Thus, applying the method of proof of Theorem
  \ref{sec:main-theorem}, each degenerate eigenvalue of $L_{A,Q}$ of
  multiplicity $M$ gives rise to $M$ branches of eigenpairs of
  $H_\eps$. The cluster of $M$ distinct eigenvalues of $H_\eps$ are
  within an interval of size $\cO(\eps^3)$ about $E_*+\eps^2 e_{A,Q}$.
  The $j^\textrm{th}$ eigen-branch satisfies the error estimates
  \begin{align}
    \label{eq:123}
    \Big \| u_\eps^{(j)}(\cdot;\mu_\eps^{(j)}) - \sum_{n=0}^N \eps^n
    U_n^{(j)}(\cdot,\eps \cdot)
    \Big \|_{H^2(\real^d)} &\leq \eps^{N-1} C,\nn\\
    \Big | \mu_\eps^{(j)} - E_* - \eps^2 e_{A,Q} - \sum_{n=3}^N \eps^n \mu_n^{(j)}
    \Big | &\leq \eps^{N+1} C,
  \end{align}
  for $j = 1, 2, \ldots, M$, all $N \geq 4$ and some constant $C>0$,
  which is independent of $\eps$.  This behavior is shown in
  Fig.~\ref{fig:multiple_degenerate} where an eigenvalue of
  multiplicity three bifurcates from the band edge.
\end{rem}

\section{Homogenization and Multi-scale Expansion}
\label{sec:homog-multi-scale}

We derive a formal asymptotic expansion for the bound state that
bifurcates from the band edge into a gap.  The results of these
calculations will be used as an ansatz in the next section \ref{sec:proof-theorem-}
 to
rigorously prove existence and error estimates.

We assume that $u_\eps(\x;\mu_\eps)$ satisfies eq.\ (\ref{eq:8})
\begin{equation}
  \label{eq:12}
  [-\Delta + V(\x) + \eps^2Q(\eps \x)] u_\eps = \mu_\eps u_\eps,
\end{equation}
and expand it in an asymptotic series as follows
\begin{equation}
  \label{eq:3}
  u_\eps(\x;\mu_\eps) = U_\epsilon(\x,\y) = \sum_{n=0}^\infty \eps^n
  U_n(\x,\y), \quad \mu_\eps 
  = E_* + \sum_{n=1}^\infty \eps^n \mu_n,
\end{equation}
where $\y = \eps \x$ is the slow variable.  Treating $\x$ and $\y$ as
independent variables, equation \eqref{eq:12} then takes the form
\begin{equation}
  \label{eq:4}
  \left[\ -\left(\ \nabla_\x + \epsilon\nabla_\y\ \right)^2  + V(\x)
    + \eps^2Q(\y)\ \right]\ U_\eps(\x,\y) = \mu_\eps U_\eps(\x,\y)\
  . 
\end{equation}
We seek a solution $U_\eps(\x,\y)$ which is periodic in the fast variable, $\x$,
 and localized in the slow variable, $\y$. Specifically, we assume $U_n(\x,\y) \in L^\infty_{\textrm{per}}
(\Omega;H^2(\real^d))$.
Inserting (\ref{eq:3}) and (\ref{eq:4}) into Eq.\ (\ref{eq:12}) and
equating like powers of $\eps$ we find
\begin{align}
  \label{eq:5}
  \bigo(\eps^0): \quad L_* U_0 \equiv&~ \left[ -\Delta_\x + V(\x) -
    E_* \right] U_0 = 0,
  \\
  \label{eq:6}
  \bigo(\eps^1): \quad L_* U_1 =&~  2 \nabla_\x \cdot \nabla_\y U_0
  + \mu_1 U_0, \\
  \label{eq:7}
  \bigo(\eps^2): \quad L_* U_2 =&~  2 \nabla_\x
  \cdot \nabla_\y U_1 + \mu_1 U_1 - [-\Delta_\y + Q(y) - \mu_2] U_0 , \\
  \nonumber & \vdots \\
  \label{eq:23}
  \bigo(\eps^n): \quad L_* U_n =&~2\nabla_\x \cdot \nabla_\y
  U_{n-1} + \mu_1 U_{n-1} - [-\Delta_\y + Q(y) - \mu_2] U_{n-2} \\
  \nonumber
  &~ + \sum_{j=3}^{n-1} \mu_j  U_{n-j}\ +\ \mu_n U_0, \quad n \geq 3. \\
  \nonumber & \vdots
\end{align}
Viewed as a system of partial differential equations for functions of
the fast variable $\x$, depending on a parameter $\y$, each equation
in this hierarchy is of the form $L_* U = \mathcal{G}(\x)$ where
$\mathcal{G}(\x)$ has the same symmetry as $w(\x)$, the band edge state (see \eqref{eq:52}),  with period cell
$\Omega$.  To solve these equations, we make repeated use of the
following two solvability criteria based on the Fredholm alternative
applied to the self-adjoint operators $L_*$ and $L_{A,Q}$ with
$\ker(L_*) = \textrm{span}\{ w \} \subset L^2_\textrm{per}(\Omega)$
and $\ker(L_{A,Q}) = \textrm{span}\{ F_{A,Q} \} \subset
L^2(\real^d_\y)$, respectively:
\begin{proposition}
  \label{sec:homog-multi-scale-2}
  Let $\mathcal{G} \in L^2_\textrm{per}(\Omega)$, then $L_* U =
  \mathcal{G}$ has an $H^2_{\textrm{per}}(\Omega)$ solution if and
  only if
  \begin{equation}
    \label{eq:112}
    \langle w, \mathcal{G} \rangle_{L^2(\Omega)} = 0.
  \end{equation}
\end{proposition}
\begin{rem}
  If $\kv_* \ne 0$, then $L^2_\textrm{per}(\Omega)$ and
  $H^2_\textrm{per}(\Omega)$ are replaced by function spaces with the
  same symmetry as $w(\x)$, $L^2_\textrm{symm}(\Omega)$ and
  $H^2_\textrm{symm}(\Omega)$.  See Remark
  \ref{sec:spectr-theory-peri-1}.
\end{rem}
\begin{proposition}
  \label{sec:homog-multi-scale-1}
  Let $\mathcal{H} \in L^2(\real^d)$, then $L_{A,Q} F = \mathcal{H}$
  has a solution $F \in H^2(\R^d_\y)$ if and only if
  \begin{equation}
    \label{eq:113}
    \langle F_{A,Q}, \mathcal{H} \rangle_{L^2(\real^d)} = 0 .
  \end{equation}
\end{proposition}

\subsection{$\bigo(\eps^0)$ Equation}
\label{sec:lead-order-equat}

From H2, there exists a unique, real, bounded eigenfunction $w \in
H^2_\textrm{per}(\Omega)$ and a simple eigenvalue $E_*$ that satisfy
\begin{equation}
  \label{eq:24}
  L_* w = [-\Delta_\x + V(\x) - E_*] w = 0, \quad \| w
  \|_{L^2(\Omega)} = 1,
\end{equation}
so that the general solution to Eq.\ \eqref{eq:5} has the multiscale
representation
\begin{equation}
  \label{eq:25}
  U_0(\x,\y) = w(\x) F_0(\y) ,
\end{equation}
for some $F_0(\y) \in H^2(\real^d)$ that will be determined at higher
order.

\subsection{$\bigo(\eps^1)$ Equation}
\label{sec:first-order-equation}

Applying Prop.~\ref{sec:homog-multi-scale-2} to eq.~\eqref{eq:6} gives
the solvability condition
\begin{equation}
  \label{eq:9}
  2 \partial_{y_j} F_0 \innero{w}{\partial_{x_j} w} + 
  \mu_1 F_0 = 0 .
\end{equation}
Since the integrand in the first term, being the derivative of the symmetric function $w^2/2$, integrates to zero,
\begin{equation}
  \label{eq:10}
  \mu_1 = 0.
\end{equation}
Therefore, the general solution for $U_1$ consists of a homogeneous
and particular solution
\begin{equation}
  \label{eq:15}
  U_1(\x,\y) = w(\x)F_1(\y) + 2 \partial_{y_j} F_0(\y)
  L_*^{-1}\{\partial_{x_j} w\}(\x) .
\end{equation}
where $F_1 \in H^2(\real^d)$ is to be determined at higher order.
\begin{rem}
  \label{sec:bigoeps1-equation}
  For $d = 1$, the general solution is
  \begin{equation}
    \label{eq:20}
    U_1(x,y) = w(x) F_1(y) + \partial_y F_0(y)  w(x) \left(- x +
      \frac{\int_0^x \frac{dx'}{w(x')^2}}{\int_0^1
        \frac{dx'}{w(x')^2}} \right) .
  \end{equation}
\end{rem}

\subsection{$\bigo(\eps^2)$ Equation}
\label{sec:second-order-equat}

Inserting the expressions \eqref{eq:10} and \eqref{eq:15} into
Eq.~\eqref{eq:7} yields
\begin{equation}
  \label{eq:11}
  L_* U_2 = 2 \partial_{y_j} F_1 \partial_{x_j} w - \mathcal{L}[F_0],
\end{equation}
where the linear operator $\mathcal{L}[G]$ for $G \in H^2(\real^d)$
is
\begin{equation}
  \label{eq:31}
  \begin{split}
    \mathcal{L}[G](\x,\y) = &~-4 \partial_{x_j} L_*^{-1}
    \{\partial_{x_l} w \}(\x) \partial_{y_j} \partial_{y_l}
    G(\y)  \\
    &~ + w(\x) [-\Delta_\y + Q(\y) - \mu_2]G(\y).
  \end{split}
\end{equation}
\begin{definition}
  Define the operator $L_{A,Q}: H^2(\real^d) \to L^2(\real^d)$ by
  \begin{equation}
    \label{eq:22}
    L_{A,Q} G(\y) \equiv \innero{w(\cdot)}{\mathcal{L}[G](\cdot,\y)} 
    =  [- \nabla_\y \cdot A \nabla_\y + Q(\y) - e_{A,Q}] G(\y), 
  \end{equation}
  where $e_{A,Q}$ is the simple eigenvalue associated with the
  eigenfunction $F_{A,Q}(\y)$ in hypothesis H3 and
  \begin{equation}
    \label{eq:26}
    A_{jl} \equiv \delta_{jl} - 4 \innero{\partial_{x_j}
      w}{L_*^{-1}\{ \partial_{x_l} w\}}. 
  \end{equation}
\end{definition}

\begin{proposition}
  \label{sec:bigoeps2-equation}
  \begin{equation}
    \label{eq:2}
    A_{jl} = \frac{1}{8\pi^2} \partial_{k_j} \partial_{k_l} E_{b_*}(\kv_*) .
  \end{equation}
\end{proposition}
We give the proof in appendix \ref{sec:summary-discussion-1}; see also
\cite{Ashcroft-Mermin:76}.

Applying Prop.~\ref{sec:homog-multi-scale-2} to Eq.\ \eqref{eq:11}
gives
\begin{equation}
  \label{eq:32}
  \innero{w(\cdot)}{\mathcal{L}[F_0](\cdot,\y)} = 0 ~ \Leftrightarrow ~ 
  L_{A,Q} F_0 = 0, \quad \mu_2 = e_{A,Q},
\end{equation}
is the effective, homogenized equation for Eq.~\eqref{eq:8} with the
effective mass tensor $A$.  We have assumed in H3 the existence of the
eigenpair $F_{A,Q} \in H^2(\real^d)$ and $e_{A,Q} \in \real \setminus
\{0\}$.  Thus, $F_0(\y) = F_{A,Q}(\y)$.

The general solution for $U_2$ consists of a homogeneous and
particular solution
\begin{equation}
  \label{eq:27}
    U_2(\x,\y) = w(\x)F_2(\y) + 2 \partial_{y_j} F_1(\y)
    L_*^{-1}\{\partial_{x_j} w\}(\x) + L_*^{-1} \{
    \mathcal{L}[F_{A,Q}](\cdot,\y) \} (\x).
\end{equation}

\subsection{$\bigo(\eps^3)$ Equation}
\label{sec:third-order-equation}

Inserting Eqs.~\eqref{eq:10}, \eqref{eq:15}, and \eqref{eq:27} into
equation \eqref{eq:23} with $n = 3$ gives
\begin{equation}
  \label{eq:13}
  \begin{split}
    L_* U_3 = & ~2 \partial_{y_j} F_2 \partial_{x_j} w - \mathcal{L}[F_1] -
    \mathcal{H}_3 + \mu_3 w F_{A,Q} ,
  \end{split}
\end{equation}
where $\mathcal{H}_3$ is known 
\begin{equation}
  \label{eq:14}
  \begin{split}
    \mathcal{H}_3(\x,\y) = &~ -2\partial_{x_j} \partial_{y_j} L_*^{-1}
    \{
    \mathcal{L}[F_{A,Q}](\cdot,\y)\}(\x) + \\
    &~2 L_*^{-1} \{\partial_{x_j} w\}(\x) [-\Delta_\y + Q(\y) -
    e_{A,Q}]\partial_{y_j} F(\y) .
  \end{split}
\end{equation}
By Prop.~\ref{sec:homog-multi-scale-2}, Eq.~\eqref{eq:13} is solvable
if and only if
\begin{equation}
  \label{eq:16}
  L_{A,Q} F_1 = - \innero{w(\cdot)}{\mathcal{H}_3(\cdot,\y)} + 
  \mu_3 F_{A,Q} .
\end{equation}
By Prop.~\ref{sec:homog-multi-scale-1}, Eq.~\eqref{eq:16} has a
solution if and only if
\begin{equation}
  \label{eq:17}
  \mu_3 = \innerr{F_{A,Q}(\cdot)}{\innero{w(\circ)}{\mathcal{H}_3(\circ,\cdot)}} .
\end{equation}
We can now write $F_1$ in terms of $F_{A,Q}$ as
\begin{equation}
  \label{eq:18}
  F_1(\y) = L_{A,Q}^{-1} \left\{ -
    \innero{w(\circ)}{\mathcal{H}_3(\circ,\cdot)} + \mu_3 F_{A,Q}(\cdot) \right
  \}(\y) .
\end{equation}
With this choice of $F_1$, eq.~\eqref{eq:13} is solvable and its
general solution is
\begin{equation}
  \label{eq:21}
  \begin{split}
    U_3(\x,\y) =&~ w(\x) F_3(\y) + 2 \partial_{y_j} F_2(\y) L_*^{-1}
    \{ \partial_{x_j} w \}(\x) - \\
    &~ L_*^{-1} \Big \{ \mathcal{L}[F_1](\cdot,\y) +
    \mathcal{H}_3(\cdot,\y) - \mu_3 w(\cdot) F_{A,Q}(\y) \Big
    \}(\x),
  \end{split}
\end{equation}
where $F_3(\y)$ is to be determined.  Note also that $F_2(\y)$,
introduced at $\mathcal{O}(\eps^2)$, is to be determined.

\subsection{$\big(\eps^n)$ Order Equation}
\label{sec:ntextrmth-order-equa}

Continuing the expansion to arbitrary $n \geq 4$ from
Eq.~\eqref{eq:23} we have
\begin{equation}
  \label{eq:28}
  L_* U_n = 2 \partial_{y_j} F_{n-1} \partial_{x_j} w - \mathcal{L}[F_{n-2}] -
  \mathcal{H}_n + \mu_n w F_{A,Q} ,
\end{equation}
where $\mathcal{H}_n$ is completely determined by all the lower order
solutions $U_l$, $l \leq n-3$
\begin{equation}
  \label{eq:29}
  \begin{split}
    \mathcal{H}_n(\x,\y) = &~2\partial_{x_j} \partial_{y_j} L_*^{-1}
    \Big \{ \mathcal{L}[F_{n-3}](\cdot,\y) + \mathcal{H}_{n-1}(\cdot,\y) -
    \mu_{n-1} w(\cdot) F_{A,Q}(\y) \Big \}(\x) \\
    &~- \sum_{l=3}^{n-1}\mu_l U_{n-l} .
  \end{split}
\end{equation}
By Prop.~\ref{sec:homog-multi-scale-2}, eq.~\eqref{eq:28} is solvable
if and only if
\begin{equation}
  \label{eq:30}
  L_{A,Q} F_{n-2} = - \innero{w(\cdot)}{\mathcal{H}_n(\cdot,\y)} +
  \mu_n F_{A,Q} .
\end{equation}
Furthermore, by Prop.~\ref{sec:homog-multi-scale-1}, eq.~\eqref{eq:30}
is solvable if and only if
\begin{equation}
  \label{eq:33}
  \mu_n = \innerr{F_{A,Q}(\cdot)}{\innero{w(\circ)}{\mathcal{H}_n(\circ,\cdot)}}, 
\end{equation}
With this choice of $\mu_n$, $F_{n-2}$ is given by
\begin{equation}
  \label{eq:34}
  F_{n-2}(\y) = L_{A,Q}^{-1} \left\{ -
    \innero{w(\circ)}{\mathcal{H}_n(\circ,\cdot)} + \mu_n F_{A,Q}(\cdot) \right
  \}(\y) . 
\end{equation}
Finally, with this choice of $F_{n-2}$, $U_n(\x,\y)$ is given by
\begin{equation}
  \label{eq:35}
  \begin{split}
    U_n(\x,\y) =&~ w(\x) F_n(\y) + 2 \partial_{y_j} F_{n-1}(\y) L_*^{-1}
    \{ \partial_{x_j} w \}(\x) - \\
    &~ L_*^{-1} \Big \{ \mathcal{L}[F_{n-2}](\cdot,\y) +
    \mathcal{H}_n(\cdot,\y) - \mu_n w(\cdot) F_{A,Q}(\y) \Big
    \}(\x).
  \end{split}
\end{equation}
Thus we have:
\begin{proposition}
  \label{sec:bigepsn-order-equat}
  The first $N$ equations \eqref{eq:5}, \eqref{eq:6}, \eqref{eq:7},
  $\ldots$, \eqref{eq:23} are solvable with solutions $U_n(\x,\y) \in
  H^2_{\textrm{per}}(\Omega; H^2(\real^d))$, $n = 0, 1, \ldots, N$
  uniquely determined up to the two arbitrary slowly varying functions
  $F_{N}(\y)$, $F_{N-1}(\y) \in H^2(\real^d)$ for $N \geq 1$.  These
  functions are the slowly varying envelopes of the homogeneous
  solutions to the $N-1^{\textrm{th}}$ and $N^{\textrm{th}}$ order
  equations.  Moreover,
  \begin{equation}
    \label{eq:46}
    U_\eps^{(N)}(\x,\y) \equiv \sum_{n=0}^N \eps^n U_n(\x,\y), \quad
    \mu_\eps^{(N)} \equiv E_* + \eps^2 e_{A,Q} + \sum_{n=3}^N \eps^n \mu_n,
  \end{equation}
  with the particular choice $F_{n-1} = F_{n} \equiv 0$, is an
  approximate solution for the eigenvalue problem \eqref{eq:12}
  (equivalently \eqref{eq:8}) with error formally of order
  $\eps^{N+1}$.
\end{proposition}

\begin{rem}
  The multi-scale form of the approximate eigenfunction given in
  Prop.~\ref{sec:bigepsn-order-equat} is used as a ``trial function'' in Appendix
  \ref{sec:vari-exist-proof} to give a ``quick'' variational existence
  proof for defect modes bifurcating from the lowest band edge.  We
  also show that a two term approximation (leading order homogenized solution plus 
  first nontrivial correction) yields a better estimate for
  the energy than the one-term approximation (leading order homogenized solution).
\end{rem}

\section{Proof of Theorem \ref{sec:main-theorem},\ Bounds on $u(\cdot,\mu)-U_\eps^{(N)},\ \mu-\mu_\eps^{(N)}$}
\label{sec:proof-theorem-}

To prove Theorem~\ref{sec:main-theorem}, we introduce the corrections
$\err(\x)$ and $\Upsilon^\eps$ to the approximate solution displayed
in eq.~\eqref{eq:46} through
\begin{equation}
  \label{eq:42}
  u(\x;\mu) \equiv U_\eps^{(N)}(\x,\eps \x) + \eps^{N-1} \err (\x),
  \quad \mu \equiv \mu_\eps^{(N)} + 
  \eps^{N+1} \Upsilon^\eps ,
\end{equation}

Then the error $\err$  satisfies the equation:
\begin{equation}
  \label{eq:19}
    [L_* - \eps^2 e_{A,Q} + \eps^2 Q(\eps \x)] \err(\x) = \eps^2
    R^\eps[\err, \Upsilon^\eps](\x),
\end{equation}
where
\begin{equation}
  \label{eq:37}
  \begin{split}
    R^\eps[\err, \Upsilon^\eps] = &~\Upsilon^\eps U_0 + 2\nabla_\x
    \cdot \nabla_\y U_N - Q U_{N-1} + \sum_{n=1}^N \mu_{n+1} U_{N-n}
    - \eps Q U_N \\
    &~ + \err \left( \sum_{n=1}^{N-2} \eps^n \mu_{n+2} + \eps^{N-1}
      \Upsilon^\eps
    \right ) + \Upsilon^\eps \sum_{n=1}^{N}\eps^n U_n \\
    &~ + \sum_{n=1}^{N-1} \eps^{n} \sum_{m=n}^N \mu_{m+1}
    U_{N+n-m} \\
    = &~ \Upsilon^\eps U_0 + 2\nabla_\x \cdot \nabla_\y U_N - Q
    U_{N-1} + \sum_{n=1}^N \mu_{n+1} U_{N-n} +
    \bigo(\eps) \\
    \equiv &~ \Upsilon^\eps w F_{A,Q} + U^\# + \mathcal{O}(\eps),
    \\
    U^\# \equiv &~2 \nabla_x \cdot \nabla_y U_N = Q U_{N-1} +
    \sum_{n=1}^N \mu_{n+1} U_{N-n} . 
  \end{split}
\end{equation}

The leading order multi-scale approximation $U_0(\x,\eps\x) =
F_{A,Q}(\eps \x) w(\x)$ in the ansatz Eq.~\eqref{eq:42} with $w(\x) =
e^{2\pi i \kv_* \cdot \x} p_{b_*}(\x;\kv_*)$ the band edge
eigenfunction suggests that the dominant contribution to the frequency
content of $\err$ will be near the band edge $E_* = E_{b_*}(\kv_*)$.
Therefore, it is natural to decompose $\err$ into Bloch eigenfunctions
\begin{equation}
  \label{eq:65}
  \err(\x) = \sum_{b=0}^\infty \int_{\Omega^*}  p_b(\x;\kv) e^{2\pi i \kv
    \cdot \x} \mathcal{T}_b \{ \err \}(\kv) \, d\kv,
\end{equation}
with associated energies or frequencies $E_b(\kv)$ where $\kv$ varies
in the Brillouin zone $\Omega^*$.  Moreover, we introduce a spectral
localization of $\err$ into frequencies ``near'' the band edge and
``far'' from the band edge
\begin{equation}
  \label{eq:40}
  \begin{split}
    \err (\x) &= \errl(\x) + \errh(\x) = \mathcal{T}^{-1}
    \errtl(\x;\cdot) + \mathcal{T}^{-1} \errth(\x;\cdot), \\
    \errtl(\x;\kv) & = \Pi_{\textrm{near}} \err(\x) \equiv \chi(| \kv
    - \kv_* | < \eps^r) \,
    \mathcal{T}_{b_*} \{ \err \} ( \kv ) \, p_{b_*}(\x; \kv) ,  \\
    \errth(\x;\kv) & = \Pi_{\textrm{far}} \err(\x) \equiv
    \sum_{b=0}^\infty \chi(| \kv - \kv_* | \geq \eps^r \delta_{b_*,b}) \,
    \mathcal{T}_b \{ \err \} ( \kv ) \, p_b(\x; \kv),
  \end{split} 
\end{equation}
where $\delta_{n,m}$ is the Kronecker delta function and the indicator
functions are defined as
\begin{equation}
  \label{eq:41}
  \chi(|\kv - \kv_*| < \eps^r) \equiv 1_{\{\kv \in \Omega^*: |\kv -
    \kv_*| < \eps^r\}}(\kv), \quad 
  \chi(|\kv - \kv_*| \geq \eps^r) \equiv 1_{\{\kv \in \Omega^*: |\kv -
    \kv_* | \ge \eps^r\}}(\kv).
\end{equation}
\begin{rem}
  For our analysis near the band edge, we will use Taylor expansions
  of various quantities about $\kv = \kv_*$.  Without loss of
  generality, we will assume that $\kv_* \equiv 0$ which enables a
  notationally cleaner presentation.  See Remark
  \ref{sec:spectr-theory-peri-1}.
\end{rem}

We will use the conventions
\begin{equation}
  \label{eq:55}
  \begin{split}
    \errtl(\kv) &\equiv \innero{p_{b_*}(\cdot;\kv)}{\errtl(\cdot;\kv)}, \\[3mm]
    \errth(\kv) &\equiv \left (
      \innero{p_{b}(\cdot;\kv)}{\errth(\cdot;\kv)} \right )_{b \ge 0},
  \end{split}
\end{equation}
where $\errtl(\kv)$ is a scalar and $\errth(\kv)$ is an infinite
vector.  This decomposition was used in
\cite{dohnal_coupled_2008,dohnal_coupled-mode_2009,Ilan-Weinstein:10}.
The parameter $r$ is assumed to lie in the interval
\begin{equation}
  \label{eq:61}
  r \in (2/3,1),
\end{equation}
the choice of which will be made clear later.

We now apply the Bloch transform to Eq.~\eqref{eq:19}, project onto
the Bloch modes $p_b(\cdot;\kv)$  and use the properties \eqref{TF-conj} and \eqref{eq:116}
 to find
\begin{equation}
  \label{eq:36}
  \begin{split}
    &[E_b(\kv) - E_* - \eps^2 e_{A,Q}]\, \Tb{b}{\err}(\kv) + \eps^2
    \Tb{b}{Q(\eps \cdot) \err(\cdot)}(\kv)  \\[3mm]
    & = \eps^2
    \Tb{b}{R^\eps[\err,\Upsilon^\eps]}(\kv) , \quad
     b = 0, 1, 
    \ldots .
  \end{split}
\end{equation}
We view this as a coupled system of equations for the near and far
frequency components $\errtl (\kv)$ and $\errth (\kv)$,\ $\kv\in\Omega^*$
\begin{align}
  \label{eq:44}
  \textrm{near:} &~ \left \{ \begin{array}{l} (E_{b_*}(\kv) - E_* -
      \eps^2 e_{A,Q}) \errtl(\kv) + \\[3mm]
      \eps^2 \chi(|\kv| <
      \eps^r) \blochn{b_*}{Q(\eps \cdot) \errl (\cdot)}(\kv) \\[3mm] =
      \eps^2 \chi(|\kv| < \eps^r) \big [ -
        \blochn{b_*}{Q(\eps \cdot) \errh (\cdot)}(\kv) \\[3mm]
        + \blochn{b_*}{R^\eps[\errl + \errh,\Upsilon^\eps]}(\kv) \big
      ],
  \end{array} \right . \\[3mm]
  \label{eq:45}
  \textrm{far:} &~ \left \{ \begin{array}{l} (E_b(\kv) - E_* - \eps^2
      e_{A,Q}) \chi(|\kv| \geq \eps^r \delta_{b_*,b})
      \errthb(\kv) \\[3mm]
      + \eps^2 \chi(|\kv| \geq \eps^r \delta_{b_*,b})
      \blochn{b}{Q(\eps
        \cdot) \errh (\cdot)}(\kv) \\[3mm]
      = \eps^2 \chi(|\kv| \geq \eps^r \delta_{b_*,b}) \Big [
      -\blochn{b}{Q(\eps \cdot) \errl (\cdot)}(\kv) + \\[3mm]
      \blochn{b}{R^\eps[\errl + \errh,\Upsilon^\eps]}(\kv) \Big ],
      \quad b = 0, 1, 2, \ldots .
  \end{array} \right .
\end{align}

\subsection{Lyapunov-Schmidt Reduction}
\label{sec:clos-resc-low}

In this section, we derive a functional representation of the far
frequency components in terms of the near frequency components with an
associated estimate.  After insertion into the near frequency equation,
a closed system is obtained.

We use the implicit function theorem to solve the far frequency
equations.  To this end, we observe the following inequalities due to
the definiteness of the matrix $A_{jl} =
\frac{1}{8\pi^2}\partial_{k_j} \partial_{k_l} E_{b_*}(0)$\ (see hypothesis H2):
\begin{equation}
  \label{eq:38}
  \begin{split}
    &|E_{b_*}(\kv) - E_* - \eps^2 e_{A,Q}| = |A_{jl}k_j k_l +
    \bigo(|\kv|^3)| \geq C \eps^{2r} > 0, \quad
    \eps^{2r} \leq |\kv|, \quad \kv \in \Omega^*,   \\[3mm]
    &|E_b(\kv) - E_* - \eps^2 e_{A,Q}| \geq C > 0, \quad | b - b_* |
    \geq 1.
  \end{split}
\end{equation}
We now have the following existence result
\begin{proposition}
  \label{sec:clos-resc-low-2}
  There exists $\eps_0 > 0$ such that for $0 < \eps < \eps_0$, there
  is a mapping $(\psi,\upsilon,\eps) \to \errth[\psi,\upsilon]$,
  $\errth: L^2(\real^d) \times \real_{\upsilon} \times \real_\eps \to
  \mathcal{X}^2$ such that $\errth[\errl, \Upsilon^\eps]$ is the
  unique solution to the far frequency Eqs.~\eqref{eq:45}.  The
  mapping $\errth$ is $C^1$ with respect to $\psi$ and $\upsilon$.
  The solution of Eqs.~\eqref{eq:45} satisfies the estimate
  \begin{equation}
    \label{eq:71}
    \begin{split}
      \left \| \errh[\errl,\Upsilon^\eps] \right \|_{H^2(\real^d)} &\leq
      C\left \| \errth[\errl,\Upsilon^\eps] \right \|_{\mathcal{X}^2}
      \\
      &\leq C \eps^{2 - 
        2r} \left( 1 + \Upsilon^\eps + (1 + \eps^{N-1} \Upsilon^\eps)\| \errl 
        \|_{L^2(\real^d)} \right) ,
    \end{split}
  \end{equation}
  for $0 < r < 1$.
\end{proposition}
\begin{proof}
  Since Eq.~\eqref{eq:45} is supported on frequencies away from $E_*$,
  we can divide it by $E_b(\kv) - E_* - \eps^2 e_{A,Q}$.  This
  suggests studying the equivalent equation $\tilde{G} = 0$ where
  $\tilde{G}$ has components
  \begin{equation}
    \label{eq:73}
    \begin{split}
      \tilde{G}_b \left [\tilde{\phi}, \eps, \psi,\upsilon \right
      ](\kv) \equiv &~ \chi(|\kv| \geq \eps^r \delta_{b_*,b})
      \bigg \{ \tilde{\phi}_b(\kv) +
      \eps^2 [ E_b(\kv) - E_* - \eps^2 e_{A,Q}]^{-1} \\
      \times \Big [\blochn{b}{Q(\eps \cdot) \phi(\cdot)}(\kv)& +
      \blochn{b}{Q(\eps \cdot) \psi(\cdot)}(\kv) -\blochn{b}{R[\psi +
        \phi,\upsilon]}(\kv) \Big ] \bigg \}, \quad b \ge 0 .
    \end{split}
  \end{equation}
  Any function $\tilde{\phi} \in \mathcal{X}^2$ satisfying
  $\tilde{G}[\tilde{\phi}, \eps, \psi, \upsilon] = 0$ is a solution of
  the far equations \eqref{eq:45} with $\psi = \errl$, $\upsilon
    = \Upsilon^\eps$.  $\tilde{G}$ is a continuous map $\mathcal{X}^2
  \times \real \times \real \to \mathcal{X}^0$ and $C^1$ with respect
  to $\psi$ and $\upsilon$ satisfying the estimate
  \begin{equation}
    \label{eq:74}
    \| \tilde{G}[\tilde{\phi},\eps,\psi,\upsilon]  \|_{\mathcal{X}^0}
    \leq C [1 + \upsilon + (1 + 
    \eps^{N+1-2r} \upsilon) \| \tilde{\phi}
    \|_{\mathcal{X}^0} + \eps^{2 - 2r}(1 + \eps^{N-1} \upsilon) \|
    \psi \|_{L^2(\real^d)} ] < \infty .
  \end{equation}
  Note that
  \begin{equation}
    \label{eq:75}
    \tilde{G}[0,0,\psi,\upsilon] = 0.
  \end{equation}
  The Proposition follows from the implicit function theorem
  \cite{nirenberg_topics_2001} if we can show that $D_{\tilde{\phi}} G
  [0,0,\psi,\upsilon]$ is invertible.  We have
  \begin{equation}
    \label{eq:76}
    \begin{split}
      \left( D_{\tilde{\phi}}
        \tilde{G}[\tilde{\phi},\eps,\psi,\upsilon] \tilde{f}
      \right)_b& (\kv)
      = \chi(|\kv| \geq \eps^r \delta_{b_*,b}) \tilde{f}_b(\kv) \\
      + \eps^2 \chi(|\kv| \geq \eps^r \delta_{b_*,b})
      &\frac{\blochn{b}{Q(\eps \cdot) f(\cdot)}(\kv) -
        \tilde{f}_b(\kv) \left (\sum_{l=1}^{N-2} \eps^l \mu_{l+2} +
          \eps^{N-1} \upsilon \right)}{E_b(\kv) - E_* - \eps^2
        e_{A,Q}} .
    \end{split}
  \end{equation}
  Therefore, $D_{\tilde{\phi}} \tilde{G}[0,0,\psi,\upsilon] = I$ is
  invertible.  Note that we use the fact that $0 < r < 1$ to conclude
  that $\lim_{\eps \to 0} \eps^2 \chi(|\kv| \geq \eps^r
  \delta_{b_*,b})/[E_b(\kv) - E_* - \eps^2 e_{A,Q}] \equiv 0$.  The
  implicit function theorem implies that there exists $\eps_0 > 0$ and
  a unique $\errth[\errl,\Upsilon^\eps] \in \mathcal{X}^2$ satisfying
  \begin{equation}
    \label{eq:77}
    \tilde{G}\left[\errth[\errl,\Upsilon^\eps],\eps,\errl,\Upsilon^\eps
    \right] = 0, 
  \end{equation}
  for $0 < \eps < \eps_0$.

  Equation \eqref{eq:77} is equivalent to
  \begin{equation}
    \label{eq:56}
    \begin{split}
      \left ( \errth \right )_b(\kv) = - \eps^2 \frac{\blochn{b}{
          \Pi_{\textrm{far}} \Big [ Q(\eps \cdot) 
          \errh(\cdot) - R^\eps[\errl + \errh,\Upsilon^\eps](\cdot) \Big ]}
        (\kv)}{E_b(\kv) - E_* - \eps^2 e_{A,Q}} .
    \end{split}
  \end{equation}
  We now demonstrate the inequality in Eq.~\eqref{eq:71}.  Using
  \eqref{eq:69}, \eqref{eq:45} and the invertibility of $L_* - \eps^2
  e_{A,Q}$ to obtain 
  \begin{equation}
    \label{eq:72}
    \begin{split}
      \| \errh \|_{H^2(\real^d)} &\leq C \left \| \errth \right \|_{\mathcal{X}^2} \\
      &= \eps^2 C \left \| \frac{\chi(|\cdot| \geq \eps^r
          \delta_{b_*,b}) \blochn{b}{ Q(\eps \circ) \errh(\circ) -
            R^\eps[\err,\Upsilon^\eps](\circ)}(\cdot)}{ E_b(\cdot) -
          E_* - \eps^2e_{A,Q}} \right \|_{\mathcal{X}^2} \\
      &\leq \eps^{2 - 2r} C \left \| \frac{(1+b^{2/d})
          \blochn{b}{Q(\eps \circ) \errh(\circ) -
            R^\eps[\err,\Upsilon^\eps](\circ)}(\cdot)}{1 + b^{2/d}}
      \right \|_{\mathcal{X}^0} \\
      & \leq \eps^{2-2r}C \left \| Q(\eps \cdot) \errh(\cdot) -
        R^\eps[\err,\Upsilon^\eps](\cdot) \right \|_{L^2(\real^d)} \\
      &\leq \eps^{2-2r}C \left [ (1 + \eps^{N-1} \Upsilon^\eps) \left
          (\| \errh(\cdot) \|_{L^2(\real^d)} + \| \errl(\cdot)
          \|_{L^2(\real^d)} \right ) + 1 + \Upsilon^\eps \right ] ,
  \end{split}
  \end{equation}
  where the constants $C$ are independent of $\eps$.  The third
  inequality results from the Weyl eigenvalue asymptotics
  \cite{Hor4} and the bound \eqref{eq:38}.  The last inequality
  results from direct estimation of the error terms \eqref{eq:37}.
  With $\eps$ small enough so that $\eps^{2-2r} C \leq 1/2$, we can
  subtract the term involving $\eps^{2-2r} C \| \errh
  \|_{H^2(\real^d)}$ from both sides of the inequality and then divide
  by $1 - \eps^{2-2r} C(1 + \eps^{N-1} \Upsilon^\eps)$ to obtain the
  desired estimate \eqref{eq:71}.
\end{proof}
\begin{rem}
  Note that we do not obtain smoothness of $\errh$ in $\eps$.  When
  applying the implicit function theorem in the above proof, we did
  not use any smoothness of the map $\tilde{G}$ in $\eps$.  This is
  because of the sharp, $\eps$ dependent cutoff function $\chi(|\kv|
  \geq \eps^r)$.
\end{rem}

\begin{rem}
  The estimate for $\errh \in H^2(\real^d)$ and $\errl \in
  L^2(\real^d)$ in \eqref{eq:71} can also be proved for $\errh \in
  H^s(\real^d)$ and $\errl \in H^{s-2}(\real^d)$ for $s \geq 2$.  The
  proof for the case $s \geq 3$ involves application of the operator
  $(I + L_*)^{s/2-1}$, which can be shown to be equivalent to the
  $H^{s-2}$ norm, to Eq. \eqref{eq:19} and necessitates further
  regularity conditions on the functions $V(\x)$, $Q(\eps \x)$, and
  $U_n(\x,\eps \x)$, $n=0, 1, \ldots, N$.
\end{rem}

\subsection{Near Frequency Equation and its Scaling}
\label{sec:analys-near-freq}

We now study the near frequency equation (\ref{eq:44}) with the aid of certain
Taylor expansions for $|\kv| < \eps^r$, where we invoke our
  regularity hypothesis H1
\begin{align}
  \nonumber E_{b_*}(\kv) = &~E_* +
  \frac{1}{2} \partial_{k_j} \partial_{k_l} E_{b_*}(0) k_j k_l +
  \frac{1}{6} \partial_{k_j}\partial_{k_l} \partial_{k_m} 
  E_{b_*}(\kv') k_j k_l k_m, \\
  \label{eq:47}
  = &~E_* + A_{jl}k_j k_l + \frac{1}{6} \partial_{k_j}\partial_{k_l} \partial_{k_m}
  E_{b_*}(\kv') k_j k_l k_m, \\
  \nonumber p_{b_*}(\x;\kv) = &~p_{b_*}(\x;0) + \partial_{k_j}
  p_{b_*}(\x;\kv'') k_j \\
  \label{eq:48}
  = &~w(\x) + \partial_{k_j}
  p_{b_*}(\x;\kv'') k_j,
\end{align}
for some $|\kv'(\kv)|, |\kv''(\kv)| < \eps^r$.  Inserting these
expansions into the near frequency equation (\ref{eq:44}), we have
\begin{equation}
  \label{eq:49}
  \begin{split}
    [A_{jl} &k_j k_l - \eps^2 e_{A,Q} ]
    \errtl(\kv) \\
    &+ \eps^2 \chi(|\kv | 
    < \eps^r) \blochn{b_*}{Q(\eps \cdot) \errl ( \cdot )}(\kv)
    = \eps^2 \tilde{R}_{\textrm{near}}^\eps[\errl,\Upsilon^\eps](\kv),
  \end{split}
\end{equation}
where
\begin{equation}
  \label{eq:50}
  \begin{split}
    \tilde{R}_{\textrm{near}}^\eps[\errl,\Upsilon^\eps]&(\kv) \equiv 
    \chi(|\kv| < \eps^r) \Big [ \blochn{b_*}{R^\eps \left [ \errl +
        \errh[\errl,\Upsilon^\eps], \Upsilon^\eps \right]}(\kv)
    \\
    & ~ - \blochn{b_*}{Q(\eps \cdot) \errh[\errl,\Upsilon^\eps]
      (\cdot)}(\kv) \\
    &~ - \frac{1}{6 \eps^2} k_j k_l k_m
    \partial_{k_j} \partial_{k_l} \partial_{k_m} e_{b_*}(\kv')
    \errtl(\kv) \Big ] , \\
    =&~ \chi(|\kv| < \eps^r) \blochn{b_*}{\Upsilon^\eps w(\cdot)
      F_{A,Q}(\eps \cdot) +
      U^\#(\cdot, \eps \cdot)}(\kv) \\
    & ~+ \bigo_{\mathcal{X}^2} \left[ (\eps + \eps^{2-2r}+
      \eps^{3r-2}) \|\errtl \|_{\mathcal{X}^2}\right] ,
  \end{split}
\end{equation}
where the leading order behavior comes from the definition of $R^\eps$
in Eq.~\eqref{eq:37}.  Recall that $U^\#$ is defined in
Eq.~\eqref{eq:37}.  The terms proportional to $\eps^{3r-2}$ put a
further restriction on the exponent $r$.  In order to keep
$\tilde{R}_{\textrm{near}}^\eps$ order one, we require
\begin{equation}
2/3 < r < 1.
\label{r-constraint}
\end{equation}
In this case, the error term satisfies the estimate
\begin{equation}
  \label{eq:39}
  \begin{split}
    &\left \| \tilde{R}^\eps_{\textrm{near}}[\errl,\Upsilon^\eps] \right
    \|_{{\cal
        X}^2} \\
    & \qquad \qquad \leq C \left [ 1 + \Upsilon^\eps + (\eps +
      \eps^{2-2r} + \eps^{3r-2}) \| \errtl \|_{\mathcal{X}^2} \right ]
    .
  \end{split}
\end{equation}
\bigskip

Equation \eqref{eq:49} can be rewritten suggestively as:
\begin{equation}
  \label{eq:49a}
  \begin{split}
 [A_{jl} &\frac{k_j}{\eps}\ \frac{k_l}{\eps} -  e_{A,Q} ]
    \errtl(\kv) \\
    &+ \chi\left(\left|\frac{\kv}{\eps}\right| 
    < \eps^{r-1}\right) \blochn{b_*}{Q(\eps \cdot) \errl ( \cdot )}(\kv)
    =  \tilde{R}_{\textrm{near}}^\eps[\errl,\Upsilon^\eps](\kv),
  \end{split}
  \end{equation}

Thus, we seek a solution of Eqs.~(\ref{eq:49}), (\ref{eq:50}) in the form
\begin{align}
  \label{eq:51}
    \errtl(\kv) &= \chi(|\kv | < \eps^r)\ \frac{1}{\eps^d}\
    \widehat{\Phi}\left(\frac{\kv}{\eps}\right)\ =\
    \chi\left(\frac{|\kv |}{\eps} < \eps^{r-1}\right)\
    \frac{1}{\eps^d}\
    \widehat{\Phi}\left(\frac{\kv}{\eps}\right), \\
    \errl(\x) &= \mathcal{T}^{-1} \left \{ \chi(|\cdot | < \eps^r)\
      \frac{1}{\eps^d}\ \widehat{\Phi}\left(\frac{\cdot}{\eps}\right)
      p_{b_*}(\x;\cdot) \right \}(\x) \nn\\
    &= \mathcal{T}^{-1} \left \{ \chi\left(\frac{|\cdot |}{\eps} <
        \eps^{r-1}\right)\ \frac{1}{\eps^d}\
      \widehat{\Phi}\left(\frac{\cdot}{\eps}\right) p_{b_*}(\x;\cdot)
    \right \}(\x)\nn
\end{align}
\bigskip

\bigskip

In the next two Lemmata,  Lemma \ref{lemma:TQPsi} and \ref{lemma:TR}, 
 we express the terms of  \eqref{eq:49a}, which involve the Gelfand-Bloch transform, in terms of the classical Fourier
transform plus a remainder, estimated to be small in $\eps$.   
\bigskip

\begin{lemma}\label{lemma:TQPsi}
\begin{itemize}
\item[(A)] Assume $\errtl(\kv)$ is given by \eqref{eq:51}. Then, 
\begin{align} [A_{jl} &\frac{k_j}{\eps}\ \frac{k_l}{\eps} -  e_{A,Q} ]
  \errtl(\kv)\nn\\
  & =\ \cF_\y \left \{ \left(\ A_{jl}\D_{y_j}\D_{y_l}-e_{A,Q}\
      \right) \ \chi(|\nabla|<\eps^{r-1})\
      \Phi \right \}\left(\frac{\kv}{\eps}\right)
  \nn \end{align}
\item[(B)]
\begin{align}
  & \blochn{b_*}{Q(\eps\cdot) \errl(\cdot)} (\kv)
  \ =\
  \frac{1}{\eps^d} \cF_\y \left\{\ Q\ \chi(|\nabla|\le\eps^{r-1})\Phi\
  \right \} \left(\frac{\kv}{\eps}\right)\ +\ \cE(\kv)
\end{align}
   where
  \begin{equation}
    \label{eq:124}
    \begin{split}
      \| \cE \|_{L^2(\real^d)} \le C  \eps^s \ \| Q \|_{H^s(\real^d)} \| \Phi
      \|_{L^2(\real^d)} ,
    \end{split}
  \end{equation}
  with $s>d$ and $0<r<1$.
  \end{itemize}
\end{lemma}

\nit {\bf Proof of Lemma \ref{lemma:TQPsi}:} \ \ Recall the notation
$\mathcal{F}f=\widehat{f}$ 
for the Fourier transform 
given by \eqref{FTdef}.  By \eqref{eq:51} since $\kv$
is localized near $0$ we have, Taylor expanding $p_{b_*}(\x;\kv)$
about $\kv=0$,
\begin{align*}
  \errl(\x) &= p_{b_*}(\x;0)\ \mathcal{F}^{-1} \left \{\ \chi\left(
      \left|\cdot\right|<\eps^{r}\right) \ \frac{1}{\eps^d}\
    \widehat{\Phi}\left(\frac{\cdot}{\eps}\right)\ \right \}\ (\x)\ +\
  \cE_1(\x) \\
  &= p_{b_*}(\x;0) \chi \left( |\nabla_\y| < \eps^{r-1} \right )
  \Phi(\y) |_{\y = \eps \x} + \cE_1(\x), \\
  \| \cE_1 \|_{L^2(\R^d)} & \le C \eps^r \| \Phi \|_{L^2(\R^d)}, \quad
  C=C\left( \| \nabla_\kv p_{b_*}
    \|_{L^\infty(\Omega\times\Omega^*)^d}\ \right).
\end{align*}
  
Since $\mathcal{T}$ commutes with multiplication by a periodic
function (see \eqref{eq:116}) and since $p_{b_*}(\x;0)$ is periodic
\begin{equation}
  \label{TQPsinear}
  \begin{split}
    & \blochn{}{Q(\eps\cdot) \errl(\cdot)}(\x,\kv)\\
    &\ \ \ = p_{b_*}(\x;0)\ \blochn{}{Q(\eps\cdot) \ \chi \left(
        |\nabla_{\eps \cdot}| < \eps^{r-1} \right ) \Phi(\eps \cdot)}
    (\x,\kv) + \blochn{}{Q(\eps \cdot) \cE_1(\cdot)}(\x,\kv) .
  \end{split}
\end{equation}

By the definition of $\mathcal{T}$, \eqref{eq:43}, we have
\begin{align}
  \label{eq:TQ}
  & \blochn{}{Q(\eps\cdot) \ \chi \left( |\nabla_{\eps \cdot}| <
      \eps^{r-1} \right ) \Phi(\eps \cdot)}
  (\x,\kv)\\
  &= \mathcal{F}_\x \ \left\{ Q(\eps\cdot)\ \chi \left( |\nabla_{\eps
        \cdot}| < \eps^{r-1} \right ) \Phi(\eps \cdot) \
  \right \} \left(\kv\right)\nn\\
  &+\ \ \sum_{\mv\ne0}\ \mathcal{F}_\x \ \left\{ Q(\eps\cdot)\
    \mathcal{F}_\kv^{-1} \left \{\ \chi\left(
        \left| \circ \right|<\eps^{r}\right) \ \frac{1}{\eps^d}\
      \widehat{\Phi}\left(\frac{\circ}{\eps}\right)\ \right \} (\cdot) \
  \right \} \left(\kv+\mv\right)\ e^{2\pi i \mv\cdot \x}\nn\\
  &=\ \ \frac{1}{\eps^d} \mathcal{F}_\y \ \left\{ Q(\cdot)\ \chi \left(
      |\nabla| < \eps^{r-1} \right ) \Phi(\cdot) \
  \right \} \left(\frac{\kv}{\eps}\right)\nn\\
  &\ \ \ \ \ +\ \sum_{|\mv|\ge 1}\ \int
  \widehat{Q}\left(\frac{\lv}{\eps}\right)\ \frac{1}{\eps^d}\
  \widehat{\Phi}\left(\frac{\kv+\mv-\lv}{\eps}\right) \
  \chi(|\kv+\mv-\lv |<\eps^r)\ d\lv\ e^{2\pi i
    \mv\cdot\x}\ \label{eq:sum} .
\end{align}

To prove \eqref{eq:124} we need to estimate the sum in \eqref{eq:sum}
for $|\kv|<\eps^r$.  For such $\kv$ the sum can be estimated as
follows:
\begin{align}
  & \sum_{|\mv|\ge 1}\ \int \left|
    \widehat{Q}\left(\frac{\lv}{\eps}\right)\right|\ \frac{1}{\eps^d}\
  \left|\widehat{\Phi}\left(\frac{\kv+\mv-\lv}{\eps}\right)\right|
  \ \chi(|\kv+\mv-\lv |<\eps^r)\ d\lv\nn\\
  &\le\ \sum_{|\mv|\ge 1}\ \left(\ \int_{|\mv-\lv|\le 2\eps^r
    } \left| \widehat{Q}\left(\frac{\lv}{\eps}\right) \right|^2\
    \frac{1}{\eps^d}\ d\lv\ \right)^{1\over2}\ \left(\
    \int_{|\mv-\lv|\le 2\eps^r } \left|
      \widehat{\Phi}\left(\frac{\lv}{\eps}\right) \right|^2\
    \frac{1}{\eps^d} \ d\lv\ \right)^{1\over2}\nn\\
  &\le C \sum_{|\mv|\ge 1} \left(1 +
    \left|\frac{\mv}{\eps}\right|^2 \right)^{-\frac{s}{2}}
  \left(\int_{|\mv-\lv|\le 2\eps^r } \left| 
      \widehat{Q}\left(\frac{\lv}{\eps}\right) \right|^2 \left(1 + 
      \left|\frac{\lv}{\eps}\right|^2 \right)^s \frac{1}{\eps^d}d\lv\
  \right)^{1\over2} \| \Phi \|_{L^2(\R^d)}\nn\\ 
  &\le\ C\ \eps^s\ \sum_{|\mv|\ge1} |\mv|^{-s}\ \|Q\|_{H^s(\R^d)}\ \| \Phi
  \|_{L^2(\R^d)}\ \le C\ \eps^s\ \|Q\|_{H^s(\R^d)}\ \| \Phi \|_{L^2(\R^d)},\ \
  s>d.\label{eq:sumbound}
\end{align}

A similar calculation shows 
\begin{equation}
  \label{eq:1}
  \| \blochn{}{Q(\eps \cdot) \cE_1(\cdot)} \|_{L^2(\R^d,\Omega^*)} \le
  C \eps^{s + r} \| Q \|_{H^s(\R^d)} \| \Phi \|_{L^2(\R^d)}, \quad s >
  d, \quad 0 < r < 1.
\end{equation}

Since $\cT_{b_*}\{f\}(\kv) = \langle p_{b_*}(\cdot,\kv),
\cT\{f\}(\cdot,\kv)\rangle_{L^2(\Omega)}$, we have by
\eqref{TQPsinear}, \eqref{eq:TQ}, \eqref{eq:sum},
\eqref{eq:sumbound}, and (\ref{eq:1}) that
\begin{align}
  &\blochn{b_*}{Q(\eps\cdot) \errl(\cdot)}(\kv)\ =\ \frac{1}{\eps^d}
  \cF_\y \left\{\ Q\ \chi(|\nabla|\le\eps^{r-1})\Phi\
  \right \} \left(\frac{\kv}{\eps}\right)\ +\ \cE(\kv) ,\nn\\
  & \|\cE\|_{L^2(\R^d)} \ \le\ C \eps^s \|Q\|_{H^s(\R^d)}\ \| \Phi
  \|_{L^2(\R^d)},\ \ s>d,\ \ 0<r<1,
\end{align}
which is the assertion of Lemma \ref{lemma:TQPsi}.
\bigskip

\begin{lemma}\label{lemma:TR}
  The right hand side of eq.~(\ref{eq:49}), defined in
  eq.~(\ref{eq:50}), satisfies
  \begin{align}
    \label{eq:126}
      &\tilde{R}^\eps_{\textrm{near}}[\errl,\Upsilon^\eps]\left(\kv \right) \\
      &=
      \frac{1}{\eps^d} \chi\left(\left|\frac{\kv}{\eps}\right|< \eps^{r-1}\right) 
      \left [ \Upsilon^\eps
        \widehat{F}_{A,Q}   \left(\frac{\kv}{\eps}\right) + 
        \innero{w(\cdot)}{\widehat{U}^\#\left(\cdot,\frac{\kv}{\eps}\right)}
      \right ]+ \cS(\kv),\nn
  \end{align}
  where
  \begin{equation}
    \label{eq:127}
    \| \cS \|_{L^2(\real^d)} \le C \left(\eps + \eps^{2 - 2r} + \eps^{3r-2} \|
    \right ) \| \Phi \|_{L^2(\real^d)} .
  \end{equation}
\end{lemma}
\begin{proof}
  The proof follows in a similar manner as to Lemma \ref{lemma:TQPsi}
  by use of eqs.~(\ref{eq:37}) and (\ref{eq:39}) along with the estimates in
  Prop.~\ref{sec:clos-resc-low-2} and eq.~(\ref{eq:50}).
\end{proof}


\subsection{Solution of the Near Frequency Equation for Small $\eps$}
\label{sec:lyap-schm-reduct}

The results from the previous section enable us to complete the proof
of Theorem \ref{sec:main-theorem}.

Substitution of the expansions of Lemmata \ref{lemma:TQPsi} and
\ref{lemma:TR} into \eqref{eq:49a} and defining
\begin{equation}
  \kav = \frac{\kv}{\eps}\label{kav-def}
\end{equation}
results in
\begin{proposition}
  The near frequency component $\Phi$ satisfies the equation
  \begin{align}
 &   \cF_\y\ \{ \left(\ A_{jl}\D_{y_j}\D_{y_l}-e_{A,Q}\ \right) \ 
           \chi( |\nabla|<\eps^{r-1} )\
           \Phi \}\left(\kav\right)\nn\\
    &\ \ \  +\  
    \cF_\y\{ Q \chi(|\nabla| < \eps^{r-1}) \Phi  \}
    (\kav) 
    \ \ = \cF_\y{H}[\Phi, \Upsilon^\eps, \eps] .
 \label{FTeqn} \end{align}

  The right hand side has the following form
  \begin{equation}
    \label{eq:88}
    \begin{split}
      \cF{H}[\Phi,\Upsilon^\eps,\eps] &= \chi(|\kav| <
      \eps^{r-1}) (\eps^d
      \tilde{R}_{\textrm{near}}^\eps[\chi(|\nabla|< \eps^{r-1})
      \Phi/\eps^d,
      \Upsilon^\eps] + \tilde{R}_c[\Phi]) \\
      &= \chi(|\kav| < \eps^{r-1})[\Upsilon^\eps \widehat{F}_{A,Q}(\kav) +
      \innero{w(\cdot)}{\widehat{U}^\#(\cdot, \kav)} \\
      &\quad +
      \bigo(\eps^{r} + \eps^{2 - 2r} + \eps^{3r - 2}) ] .
    \end{split}
  \end{equation}
  We define the following operators $\chi_\eps$,
  $\overline{\chi}_\eps$ where
  \begin{equation}
    \label{eq:63}
    \chi_\eps \equiv \chi(|\nabla_\y| < \eps^{r-1}), \quad
    \overline{\chi}_\eps \equiv 1 - \chi_\eps = \chi(|\nabla_\y| \geq
    \eps^{r-1}) .
  \end{equation}
  In physical space we can write \eqref{FTeqn} as
  \begin{equation}
    \label{eq:64}
    \chi_\eps L_{A,Q} \chi_\eps \Phi = \chi_\eps
    H[\Phi,\Upsilon^\eps,\eps] .
  \end{equation}
  where
  \begin{equation}
    \label{eq:93}
    \begin{split}
      H[\Phi,\Upsilon^\eps,\eps](\y) &= 
      \Upsilon^\eps F_{A,Q}(\y) + \innero{w(\cdot)}{U^\#(\cdot,\y)}
      + h[\Phi,\Upsilon^\eps,\eps]  , \\
      \| \chi_\eps h[\Phi,\Upsilon^\eps,\eps] \|_{H^2(\real^d)} &\leq
      C(\eps^r + \eps^{2 - 2r} + \eps^{3r-2})(1 + \| \Phi
      \|_{H^2(\real^d)} ) .
    \end{split}
  \end{equation}
\end{proposition}

In order to solve Eq.~\eqref{eq:64}, we require a regularization that
guarantees the invertibility of the operator $\chi_\eps L_{A,Q}
\chi_\eps$.  Since zero is an isolated eigenvalue of $L_{A,Q}$, there
is a small disc of radius $\rho=\rho_\eps$ about zero, with boundary
$C_\rho$ such that for $\eps$ sufficiently small, $C_\rho$ encircles
$m$ eigenvalues of $\chi_\eps L_{A,Q} \chi_\eps$, counting
multiplicity, where $m$ is the multiplicity of zero as an eigenvalue
of $L_{A,Q}$.

Introduce the projection onto the spectral subspace associated with eigenvalues
of $\chi_\eps L_{A,Q} \chi_\eps$, encircled by $C_\rho$:
\begin{equation}
  \label{eq:59}
  \Pi_\eps \equiv \frac{1}{2\pi i} \int_{C_\rho} (\chi_\eps L_{A,Q} \chi_\eps -
  \lambda I)^{-1} \, d\lambda ,
\end{equation}
Note that
\begin{equation}
  \label{eq:106}
  \Pi_0 = \innerr{F_{A,Q}}{\cdot} ,
\end{equation}
projects onto the kernel of $L_{A,Q}$.

We now rewrite \eqref{eq:64} as the following 
system for $\Phi$ and $\Upsilon^\eps$:
\begin{align}
  \label{eq:104}
  \chi_\eps L_{A,Q} \chi_\eps \Phi &= \chi_\eps (I - \Pi_\eps)
  \chi_\eps H[\Phi,\Upsilon^\eps,\eps]
  \\
  \label{eq:105}
  \chi_\eps \Pi_\eps \chi_\eps H[\Phi,\Upsilon^\eps,\eps] &= 0 .
\end{align}
Any solution $(\Phi^\eps,\Upsilon^\eps)$ of \eqref{eq:104}, \eqref{eq:105} is a solution of \eqref{eq:64}

We claim that for 
$\eps$ small ~\eqref{eq:104} can be solved for $\Phi =
\Phi^\eps[\Upsilon^\eps]$ via the equivalent  nonlocal ``integral'' equation:
\begin{equation}
  \label{eq:58}
  \Phi^\eps = (\chi_\eps L_{A,Q} \chi_\eps)^{-1} \ (I - \Pi_\eps)\  
  \left(\Upsilon^\eps F_{A,Q} + \innerr{w(\cdot)}{U^\#(\cdot,\y)} +
  h[\Phi^\eps,\Upsilon^\eps,\eps] \right) .
\end{equation}
Indeed, the solution may be constructed using the iteration:
\begin{equation}
  \label{eq:60}
  \begin{split}
    \Phi^\eps_{j+1} &= (\chi_\eps L_{A,Q} \chi_\eps)^{-1} \chi_\eps (I - \Pi_\eps)
    \chi_\eps (\Upsilon^\eps F_{A,Q} + \innerr{w(\cdot)}{U^\#(\cdot,\y)} +
    h[\Phi^\eps_j,\Upsilon^\eps,\eps] ) ,\\
    \Phi^\eps_0 &= (\chi_\eps L_{A,Q} \chi_\eps)^{-1} \chi_\eps (I - \Pi_\eps)
    (\Upsilon^\eps F_{A,Q} + \innerr{w(\cdot)}{U^\#(\cdot,\y)} ).
  \end{split}
\end{equation}
By use of  \eqref{eq:59} and  \eqref{eq:93}, we have
\begin{align}
  \label{eq:62}
    &\left \| \Phi_{j+1} - \Phi_j \right \| 
    \leq 
    \left \| (\chi_\eps L_{A,Q} \chi_\eps)^{-1} \chi_\eps (I - \Pi_\eps)
      \chi_\eps (h[\Phi_j ,\Upsilon^\eps,\eps] - h[\Phi_{j-1},
      \Upsilon^\eps, \eps]) \right \|  \nn \\
    & \qquad \qquad \leq  \tau(\eps)\ \left \| \Phi_{j} - \Phi_{j-1} \right
    \|,\ \ \ \  \tau(\eps)\equiv\rho^{-1} C (\eps^r + \eps^{2-2r} +
    \eps^{3r-2}),\ \ \frac{2}{3} < r < 1\nn .
\end{align} 
Therefore, $\left \| \Phi_{j+1} - \Phi_j \right \| \le\ \tau(\eps)^j
\left \| \Phi_{1} - \Phi_0 \right \|$ and if $\eps$ satisfies the
smallness condition
$  \tau(\eps)< 1$,
the sequence $\{\Phi_j^\eps \}_{j \ge 0}$ is Cauchy in
$H^2(\real^d)$.  It therefore contains a subsequence, which is
convergent to a limit $\Phi_*^\eps \in H^2(\real^d)$.  By
$H^2(\real^d)$ continuity of the terms in the iteration (\ref{eq:60}),
one can pass to the limit in \eqref{eq:60} to obtain a solution
$\Phi_*^\eps$ which satisfies Eq.~(\ref{eq:58}).
 
This solution $\Phi_*^\eps$ is a functional of $\Upsilon^\eps$, and
appears in equation \eqref{eq:105}, which we view as an equation for
$\Upsilon^\eps$.  We write \eqref{eq:105} in the form
\begin{equation}
  \label{eq:121}
  \begin{split}
    g[\upsilon,\eps] \equiv & ~ \upsilon \chi_\eps \Pi_\eps
    F_{A,Q}(\y) + \chi_\eps \Pi_\eps 
    \innero{w(\cdot)}{U^\#(\cdot,\y)} \\
    &+ \chi_\eps \Pi_\eps
    h[\Phi[\upsilon],\upsilon,\eps] = 0 .
  \end{split}
\end{equation}
For $\eps = 0$, this equation has the solution $g[\upsilon_0,0] = 0$
with 
\begin{equation}
  \label{eq:122}
  \upsilon_0 = -
  \innerr{F_{A,Q}(\cdot)}{\innero{w(\circ)}{U^\#(\circ,\cdot)}} .
\end{equation}
The Jacobian, $D_\upsilon g[\upsilon_0,0] = 1$.  By the implicit
function theorem \cite{nirenberg_topics_2001}, for $|\eps|$
sufficiently small there exists a unique solution
$\eps\mapsto\Upsilon^\eps$ satisfying $g[\Upsilon^\eps,\eps] = 0$.
This completes the proof of Theorem \ref{sec:main-theorem}.

\appendix
 
\section{Effective Mass Tensor}
\label{sec:summary-discussion-1}

In this appendix we prove Proposition \ref{sec:bigoeps2-equation},
relating the Hessian matrix of the band dispersion function
$E_{b_*}(\kv)$ to the matrix $A$ resulting from the multiple-scale
analysis.  In addition, we prove hypothesis H2(b) under certain
conditions and the positive definiteness of $I - A$.

The solutions to the eigenvalue equation (\ref{ev-eqn}) are sought in
the form $u_b(\x;\kv) = e^{2\pi i\,\kv \cdot \x} p_b(\x;\kv)$, $\kv
\in \Omega_*$.  Then, $p_{b_*}$ and $E_{b_*}$ satisfy
\begin{equation}
  \label{eq:84}
  L_*^{(\kv)} p_{b_*} \equiv \left(- \Delta - 4\pi i \: \! \kv \cdot \nabla +
    4\pi^2 |\kv|^2   
    + V(\x) - E_{b_*}(\kv) \right) p_{b_*}(\x,\kv) =
  0, \quad \x \in \T^d ,
\end{equation}
with periodic boundary conditions $p_{b_*}(\x +
\mathbf{e}_j;\kv) = p_{b_*}(\x;\kv)$.  Taking the derivative
of eq.~(\ref{eq:84}) with respect to $k_j$ gives
\begin{equation}
  \label{eq:83}
  L_*^{(\kv)} \partial_{k_j} p_{b_*}(\x; \kv) = \left ( 4\pi i \partial_{x_j} -
    8\pi^2 k_j + \partial_{k_j} E_{b_*}(\kv)
  \right )p_{b_*}(\x;\kv) . 
\end{equation}
Evaluating eq.~(\ref{eq:83}) at $\kv = \kv_*$ and using the fact that
the kernel of $L_*^{(\kv_*)}$ is spanned by $p_{b_*}(\x;\kv_*)$, we
arrive at the solvability condition
\begin{equation}
  \label{eq:98}
  \begin{split}
    \partial_{k_j} E_{b_*}(\kv_*) = -4\pi i \innero{\left
        ( \partial_{x_j} + 2\pi i k_{*,j}
        \right ) p_{b_*}(\cdot;\kv_*)}{p_{b_*}(\cdot;\kv_*) } ,
  \end{split}
\end{equation}
for $j = 1,2, \ldots, d$.  When $\kv_* = 0$, $p_{b_*}(\x;\kv_*)$ is
real valued so that eq.~(\ref{eq:98}) simplifies to
\begin{equation}
  \label{eq:100}
  \partial_{k_j} E_{b_*}(\kv_*) = 0, \quad j = 1,2, \ldots, d .
\end{equation}
For the case $d = 1$, we also have
\begin{equation}
  \label{eq:102}
  E'_{b_*}(k_*) = 0, \quad k_* \in \{0, \pm 1/2 \}, \quad d = 1 .
\end{equation}
This result follows from properties of the Floquet discriminant $\Delta(E)$
\cite{Eastham-73}.  Briefly, for each $E$, one constructs a $2 \times
2$ fundamental matrix of solutions $M(E)$ and considers the values of
$E$ for which $M(E)$ has an eigenvalue $1$ or $-1$ corresponding to a
periodic or antiperiodic eigenvalue, respectively.  This is equivalent
to $\Delta(E) = \pm 2$ where
\begin{equation}
  \label{eq:103}
  \textrm{trace}(M(E_{b}(k))) = \Delta( E_{b}(k) ) = 2 \cos(2 \pi k(E_b)) .
\end{equation}
Differentiating this expression with respect to $k$ and evaluating at
$k = k_*$ gives
\begin{equation}
  \label{eq:108}
  \frac{d \Delta}{d E} \left( E_{b_*}(k) \right ) E_{b_*}'(k) = - 4 \pi \sin
  (2\pi k) .
\end{equation}
Since $d\Delta/dE(E) = 0$ if and only if $E$ is a double eigenvalue
(Theorem 2.3.1, \cite{Eastham-73}) and $E_{b_*}(k_*)$ is assumed
simple (hypothesis H2(a)), eq.~(\ref{eq:102}) follows.  

The above discussion proves hypothesis H2(b) at the left band edge
$\kv_* = 0$ for arbitrary $d$ and both left and right band edges $k_*
\in \{0, \pm 1/2 \}$ when $d = 1$.  It is possible for $\nabla
E_{b_*}(\kv_*) = 0$ in other cases, e.g.~separable potentials, and we
continue the discussion assuming this to be true.

It follows that
\begin{equation}
  \label{eq:99}
  \partial_{k_j} p_{b_*}(\x;\kv_*) = 
  4\pi i \big (  L_*^{(\kv_*)} \big )^{-1} \left \{ (\partial_{x_j}
    + 2 \pi i k_{*,j}) p_{b_*}(\cdot ; \kv_*) \right \}(\x) . 
\end{equation}
Differentiating eq.~(\ref{eq:83}) with respect to $k_l$ and setting
$\kv = \kv_*$ gives
\begin{equation}
  \label{eq:97}
  \begin{split}
    L_*^{(\kv_*)} \partial_{k_j} \partial_{k_l}
    &p_{b_*}(\x; \kv_*) = \\
    &-16 \pi^2 \left ( \partial_{x_l} +
      2\pi i k_{*,l} \right ) \big ( L_*^{(\kv_*)} \big )^{-1} \left \{ (\partial_{x_j}
      + 2 \pi i k_{*,j}) p_{b_*}(\cdot ; \kv_*) \right \}(\x) \\
    &-16 \pi^2 \left ( \partial_{x_j} +
      2\pi i k_{*,j} \right ) \big ( L_*^{(\kv_*)} \big )^{-1} \left \{ (\partial_{x_l}
      + 2 \pi i k_{*,l}) p_{b_*}(\cdot ; \kv_*) \right \}(\x) \\
    &+ \left ( \partial_{k_j
        k_l} E_{b_*}(\kv_*) - 8\pi^2 \delta_{jl} \right ) p_{b_*}(\x;\kv_*) .
  \end{split}
\end{equation}
Invoking the solvability condition and using $\| p_{b_*}
\|_{L^2(\Omega)} = 1$, eq.~(\ref{eq:100}) gives
\begin{equation}
  \label{eq:101}
  \begin{split}
    &\partial_{k_j k_l} E_{b_*}(\kv_*) = 8 \pi^2 \delta_{jl} \\
    &+ 16\pi^2 \innero{(\partial_{x_l} + 2\pi i k_{*,l}) \big (
      L_*^{(\kv_*)} \big )^{-1} \left \{ (\partial_{x_j} + 2 \pi i
        k_{*,j}) p_{b_*}(\circ ; \kv_*) \right \}(\cdot)
    }{p_{b_*}(\cdot;\kv_*)} \\
    &+ 16\pi^2 \innero{(\partial_{x_j} + 2\pi i k_{*,j}) \big (
      L_*^{(\kv_*)} \big )^{-1} \left \{ (\partial_{x_l} + 2 \pi i
        k_{*,l}) p_{b_*}(\circ ; \kv_*) \right \}(\cdot)
    }{p_{b_*}(\cdot;\kv_*)} .
  \end{split}
\end{equation}
Integrating by parts and using the fact that $\big( L_*^{(\kv_*)} \big
)^{-1}$ is self-adjoint, the last two terms are equal giving
\begin{equation}
  \label{eq:117}
  \begin{split}
    &\partial_{k_j k_l} E_{b_*}(\kv_*) = 8 \pi^2 \delta_{jl} \\
    &- 32\pi^2 \innero{(\partial_{x_j} + 2\pi i k_{*,j}) p_{b_*}(\cdot;\kv_*)}{\big (
      L_*^{(\kv_*)} \big )^{-1} \left \{ (\partial_{x_l} + 2 \pi i
        k_{*,l}) p_{b_*}(\circ ; \kv_*) \right \}(\cdot)} .
  \end{split}
\end{equation}
In order to identify eq.~(\ref{eq:117}) with the final result,
eq.~(\ref{eq:2}) in Prop.~\ref{sec:bigoeps2-equation}, we use the
definition of $w(\x)$ in eq.~(\ref{eq:52}) to compute
\begin{equation}
  \label{eq:118}
  \partial_{x_j} w(\x) = e^{2\pi i \kv_* \cdot \x} \left
    ( \partial_{x_j} + 2\pi i k_{*,j} \right )p_{b_*}(\x;\kv_*),
\end{equation}
which is the first term in the inner product of eq.~(\ref{eq:118}).
In addition, the identity
\begin{equation}
  \label{eq:119}
  L_* e^{2\pi i \kv_* \cdot \x} f(\x) = e^{2\pi i \kv_* \cdot \x}
  L_*^{(\kv_*)} f(\x) ,
\end{equation}
implies
\begin{equation}
  \label{eq:120}
  e^{2\pi i \kv_* \cdot \x} \big (L_*^{(\kv_*)} \big )^{-1} \left \{
    (\partial_{x_l} + 2 \pi i 
    k_{*,l}) p_{b_*}(\cdot ; \kv_*) \right \}(\x) = L_*^{-1}\left
    \{ \partial_{x_l} w(\cdot) \right \}(\x) ,
\end{equation}
and the result follows.

\section{Homogenization and Variational Analysis}
\label{sec:vari-exist-proof}

The existence of a bound state for eq.~\eqref{eq:81} bifurcating from
the lowest band edge $E_0(0)$ can be proved by showing that the
Rayleigh quotient
\begin{equation}
  \mathcal{E}[u] = \frac{\int_{\real^d} \left \{ |\nabla u(\x) |^2 +
      [V(\x) + \eps^2 Q(\eps 
      \x) - E_0(0)] |u(\x)|^2 d\x \right \} }{\int_{\real^d} |u(\x)|^2 d\x},
\end{equation}
is negative for some choice of $u \in H^1(\real^d)$
\cite{lieb_analysis_2001}.  A natural choice for $u$ is the
multi-scale expansion in eq.~\eqref{eq:3} with $\eps$ sufficiently
small.  Furthermore, a higher order, two-term trial function gives a
better approximation of the energy than the one-term trial function.
\begin{proposition}
  \label{sec:vari-exist-proof-2}
  \begin{enumerate}
  \item \textbf{Negative energy trial function:} with assumptions
    H1-H3 in section \ref{sec:main-results} and setting $E_* \equiv
    E_0(0)$, the lowest band edge, then there exists $\eps_0 > 0$ such
    that for all $0 < \eps < \eps_0$,
    \begin{equation}
      \label{eq:57}
      \mathcal{E}\left[ F_{A,Q}(\eps \circ) w(\circ) + 2 \nabla_\circ
        F_{A,Q}(\eps \circ) \cdot L_*^{-1} \left \{ \nabla w \right
        \}(\circ) \right ] < 0 .
    \end{equation}
    It follows that there exists a ground state.
  \item \textbf{Estimates of the ground state energy:} if $L_{I,Q}
    \equiv -\Delta_\y + Q(\y)$ has a simple eigenvalue $e_{I,Q} < 0$
    and corresponding eigenfunction $F_{I,Q}(\y) \in H^2(\real^d)$,
    then there exists $\eps_0 > 0$ such that for all $0 < \eps <
    \eps_0$,
  \begin{equation}
    \label{eq:109}
    \mathcal{E}\left[ F_{A,Q}(\eps \circ) w(\circ) + 2 \nabla_\circ
      F_{A,Q}(\eps \circ) \cdot L_*^{-1} \left \{ \nabla w \right
      \}(\circ) \right ] < \mathcal{E}\left[ F_{I,Q}(\eps \circ) 
      w(\circ) \right ] < 0 .
  \end{equation}
  \end{enumerate}
\end{proposition}

For the proof of Prop.~\ref{sec:vari-exist-proof-2}, we will make
repeated use of the following averaging lemma.
\begin{lemma}
  \label{sec:vari-exist-proof-1}
  Let $p(\x) = p(\x + \z)$ be periodic with fundamental period cell
  $\Omega$ and $\sum_{\z \in \integer^d} |\widehat{p}_\z| < \infty$
  where $\widehat{p}_\z$ are the Fourier series coefficients of
  $p(\x)$.  If $G \in L^1(\real^d) \cap C^n(\real^d)$,
  then
  \begin{equation}
    \label{eq:114}
    \left | \int_{\real^d} p(\x) G(\eps \x) d\x - \frac{1}{\eps^d}
      \fint_{\Omega} p(\x) d\x \int_{\real^d} G(\y) d\y \right |  \le
    C \eps^n .
  \end{equation}
\end{lemma}
\begin{proof}
  Expand $p$ in the Fourier series $p(\x) = \sum_{\z \in \integer^d}
  \widehat{p}_\z e^{2\pi i \z \cdot \x}$.  Then
  \begin{equation}
    \label{eq:80}
    \begin{split}
      \int_{\real^d} p(\x) G(\eps \x) \, d\x &= \frac{1}{\eps^d}
      \sum_{\z \in \integer^d} \widehat{p}_\z \int_{\real^d} G(\y)
      e^{2\pi i \z
        \cdot \y/\eps} \, d\y \\
      &= \frac{\widehat{p}_0}{\eps^d} \int_{\real^d}
      G(\y) \, d\y +  \frac{1}{\eps^d} \sum_{\z \ne 0} \widehat{p}_\z
      \int_{\real^d} 
      G(\y) e^{2\pi i \z \cdot \y/\eps} d\y \\
      &= \frac{1}{\eps^d} \fint_\Omega p(\x) \, d\x \int_{\real^d}
      G(\y) \, d\y + \cO(\eps^n),
    \end{split}
  \end{equation}
  where the first line is justified by the assumed absolute
  convergence of the Fourier coefficients and the second line results
  from integration by parts $n$ times.
\end{proof}

First we consider the ansatz $u(\x) = F_{I,Q}(\eps \x) w(\x)$ for
eq.~(\ref{eq:109}).  A computation and several applications of the
averaging lemma~\ref{sec:vari-exist-proof-1} give
\begin{equation}
  \label{eq:85}
  \begin{split}
    \mathcal{E}\left[ F_{I,Q}(\eps \circ) 
      w(\circ) \right ] &= \eps^2 \frac{\int_{\real^d} \left \{ |\nabla F_{I,Q}(\y) |^2 +
        Q(\y)|F_{I,Q}(\y)|^2 \right \} d\y }{\int_{\real^d}
      |F_{I,Q}(\y)|^2 d\y}  + o(\eps^{2}) \\
    &= \eps^2 e_{I,Q} +
    o(\eps^{2}) < 0,
  \end{split}
\end{equation}
for $\eps$ sufficiently small.

A similar, more involved computation for the ansatz $u(\x) =
F_{A,Q}(\eps \x) w(\x) + 2 \nabla F_{A,Q}(\eps \x) \cdot L_*^{-1}
\left \{ \nabla w \right \}(\x)$ leads to
\begin{equation}
  \label{eq:90}
  \begin{split}
    &\mathcal{E}\left[ F_{A,Q}(\eps \circ) w(\circ) + 2 \nabla
      F_{A,Q}(\eps \circ) \cdot L_*^{-1} \left \{ \nabla w \right
      \}(\circ) \right ] \\
    &= \eps^2 \frac{\int_{\real^d} \left \{ \nabla F_{A,Q}(\y) \cdot A
        \nabla F_{A,Q}(\y) + Q(\y) | F_{A,Q}(\y) |^2  \right \} d\y
    }{\int_{\real^d} |F_{A,Q}(\y)|^2 d\y} +
    o(\eps^{2}) \\
    &= \eps^2 e_{A,Q} + o(\eps^{2}) < 0,
  \end{split}
\end{equation}
for $\eps$ sufficiently small.  For a bifurcation from the lowest band
edge, the effective mass tensor $A$ is positive definite
\cite{Kirsch-Simon:87} hence the eigenvalue $e_{A,Q}$ is negative.

The proof of Prop.~\ref{sec:vari-exist-proof-2} is completed if we can
show that $e_{A,Q} < e_{I,Q}$.  For this, we use the following
proposition.
\begin{proposition}
\label{sec:effect-mass-tens}
The operator $I - A = \left ( 4 \innero{\partial_{x_j} w}{L_*^{-1}
    \{ \partial_{x_l} w \}} \right )$ is positive definite at the
lowest band edge $(b_* = 0$, $\kv_* = 0)$.
\end{proposition}
\begin{proof}
  Recall that $L_* \ge 0$ with one-dimensional $L^2(\T^d)$ kernel
  spanned by $w$.  Let $\mathbf{v} = (v_1,\ldots, v_d) \in \R^d$ be
  arbitrary.  Then
  \begin{equation}
    \label{eq:115}
    \mathbf{v} \cdot (I - A) \mathbf{v} = \innero{\mathbf{v} \cdot \nabla
      w}{L_*^{-1} \mathbf{v} \cdot \nabla w} \ge \frac{1}{E_1(0) - E_*}
    \| \mathbf{v} \cdot \nabla w \|^2 \ge C \| \mathbf{v} \|^2, \quad
    C > 0,
  \end{equation}
  where $E_1(0) - E_* > 0$ is the second eigenvalue of $L_*$ acting on
  $L^2(\T^d)$.
\end{proof}

Introducing the energy functional
\begin{equation}
  \label{eq:66}
  J_{A,Q}[g] \equiv \frac{\int_{\R^d} \left [ \nabla g(\y) \cdot A \nabla
      g(\y) + Q(\y) g^2(\y) \right ] \, d\y}{\int_{\R^d} g^2(\y) \,
    d\y} , \quad g \in L^2(\Omega) ,
\end{equation}
we observe
\begin{equation}
  \label{eq:78}
  J_{A,Q}[g] - J_{I,Q}[g] = \frac{\int_{\R^d} \left [ \nabla g(\y)
      \cdot (A - I) \nabla g(\y) \right ] \,
    d\y}{\int_{\R^d} g^2(\y) \, d\y} < 0, \quad g \in L^2(\Omega),
\end{equation}
by the negative definiteness of $A - I$ from
Prop.~\ref{sec:effect-mass-tens}.  Using eqs.~(\ref{eq:85}),
(\ref{eq:90}), and (\ref{eq:78}) we find
\begin{equation}
  \label{eq:79}
  \begin{split}
    e_{A,Q} = J_{A,Q}[F_{A,Q}] = \inf_{g \in L^2(\Omega)} J_{A,Q}[g]
    \le J_{A,Q}[F_{I,Q}] < J_{I,Q}[F_{I,Q}] = e_{I,Q} ,
  \end{split}
\end{equation}
and the proof is complete.

\bibliographystyle{siam}

\bibliography{homog-defect}
\end{document}

%% file: abbrevs.tex
\newcommand{\nc}{\newcommand}

\nc{\fig}[4]
{
  \begin{figure}[ht!]  
    \centering{\scalebox{#1}{\includegraphics*{./figures/#2.eps}}}
    \caption{#4}
    \label{fig:#3}
  \end{figure}
}

\nc{\bXYZ}{{\bf XYZ\ }}

\nc{\nn}{\nonumber} 
\nc{\nit}{\noindent}
\nc{\marginnote}[1] {\marginpar{\tiny #1}}
\nc{\SH}{Schr\"odinger}
\nc{\NLS}{nonlinear Schr\"odinger}
\nc{\ie}{\emph{i.e.\ \mbox{}}}
\nc{\eg}{\emph{e.g.\ \mbox{}}}
\nc{\per}{{l_j}}
\nc{\perx}{q}
\nc{\vol}{{\rm Vol}}

\nc{\lapx}{\dfrac{\partial^2}{\partial x^2}}
\nc{\lapy}{\dfrac{\partial^2}{\partial y^2}}
\nc{\lapxy}{\dfrac{\partial^2}{\partial xy}}
\nc{\diff}[2]{\frac{d #1}{d #2}}
\nc{\diffn}[3]{\frac{d^{#3} #1}{d {#2}^{#3}}} 
\nc{\pdiff}[2]{\frac{\partial #1}{\partial #2}} 
\nc{\pdiffn}[3]{\frac{\partial^{#3} #1}{\partial{#2}^{#3}}} 

\def\Xint#1{\mathchoice
  {\XXint\displaystyle\textstyle{#1}}%
  {\XXint\textstyle\scriptstyle{#1}}%
  {\XXint\scriptstyle\scriptscriptstyle{#1}}%
  {\XXint\scriptscriptstyle\scriptscriptstyle{#1}}%
  \!\int}
\def\XXint#1#2#3{{\setbox0=\hbox{$#1{#2#3}{\int}$} \vcenter{\hbox{$#2#3$}}\kern-.5\wd0}}
\def\dashint{\Xint-}
\nc{\Av}{\dashint_\cell}
\nc{\avg}[1]{\mbox{$\left \langle\ \!#1\!\!\ \right \rangle$}}
\nc{\intRd}{\int_{{\mathbb R}^d}}

\nc{\abs}[1] {\lvert #1 \rvert} 
\nc{\norm}[2] {{\lVert #1 \rVert}_{#2}} 
\nc{\normH}[1]{\|#1\|_{H^1}^2}
\nc{\normL}[1]{\|#1\|_2^2}
\nc{\nl}[1]{|#1|^{2\sigma}}
\nc{\nlb}[1]{#1^{2\sigma+1}}
\nc{\Linvd}{L_\delta^{-1}}
\nc{\Linvz}{L_0^{-1}}
\nc{\Linvzs}{L_0^{-2}}

\nc{\st}[1]{\stackrel{(\ref{eq:#1})}{=}}
\nc{\stt}[2]{\stackrel{(\ref{eq:#1}),(\ref{eq:#2})}{=}}

\nc{\veps}{\epsilon}
\nc{\Ue}{U_\eps}
\nc{\koe}{\frac{k}{\veps}} 

\nc{\wrat}{w_{\rm ratio}}
\nc{\aij}{A^{ij}}                  
\nc{\minv}{m_*^{-1}} 
\nc{\sqminv}{m_*^{-\frac{1}{2}}}
\nc{\G}{\gamma_{\rm ef\,\!f}}
\nc{\zetap}{{\zeta_*}}
\nc{\zetas}{{\zeta_{1*}}}
\nc{\Pe}{P_{\rm edge}}
\nc{\slope}{\dfrac{d\cP[u(\cdot;\mu)] }{d\mu}}
\nc{\moot}{\mu_2}
\nc{\mum}{\mu_{\rm min}}

\def\R{\mathbb{R}}

\def\T{\mathbb{T}}

\nc{\bk}{{\bf k}}
\nc{\bq}{{\bf q}}
\nc{\br}{{\bf r}}
\nc{\bx}{{\bf x}}
\nc{\by}{{\bf y}}
\nc{\bz}{{\bf z}}
\nc{\bl}{{\bf l}}
\nc{\bnu}{{\bf \nu}}
\nc{\bxc}{{\bf x}_c}
\nc{\bxi}{{\bold\xi}}

\nc{\cP}{{\cal P}}
\nc{\cPedge}{{\cal P}_{edge}}
\nc{\cPc}{{\cal P}_{cr}} 
\nc{\cell}{{\cal B}}
\nc{\cDD}{\Lambda} 
\nc{\cF}{{\cal F}}
\nc{\cG}{{\cal G}}
\nc{\cO}{{\cal O}}  
\nc{\cQ}{{\cal Q}}  
\nc{\cR}{{\cal R}}
\nc{\cI}{{\cal I}}
\nc{\cK}{{\cal K}}
\nc{\cL}{{\cal L}} 
\nc{\cM}{{\cal M}}
\nc{\cN}{{\cal N}} 
\nc{\cE}{{\cal E}}
\nc{\cH}{{\cal H}} 
\nc{\cS}{{\cal S}} 
\nc{\cX}{{\cal X}}
\nc{\cZ}{{\cal Z}} 
\nc{\cT}{{\cal T}} 
\nc{\order}{{\cal O}}

\newcommand{\blochn}[2]{\mathcal{T}_{#1} \left \{ #2 \right \}}
\newcommand{\innero}[2]{\left \langle #1, #2 \right \rangle_{L^2(\Omega)}}
\newcommand{\innerr}[2]{\left \langle #1, #2 \right \rangle_{L^2(\real^d)}}
\newcommand{\Tb}[2]{\mathcal{T}_{#1} \left \{ #2 \right \}}
\newcommand{\real}{\mathbb{R}}

\newcommand{\integer}{\mathbb{Z}}

\newcommand{\kv}{\mathbf{k}}

\newcommand{\Ls}{L^2_\textrm{symm}(\Omega)}
\newcommand{\mv}{\mathbf{m}}
\newcommand{\kav}{\boldsymbol{\kappa}}
\newcommand{\lv}{\mathbf{l}}
\newcommand{\x}{\mathbf{x}}
\newcommand{\y}{\mathbf{y}}
\newcommand{\z}{\mathbf{z}}
\newcommand{\err}{\Psi^\eps}
\newcommand{\errl}{\Psi^\eps_{\textrm{near}}}
\newcommand{\errh}{\Psi^\eps_{\textrm{far}}}
\newcommand{\errtl}{\tilde{\Psi}^\eps_{\textrm{near}}}
\newcommand{\errth}{\tilde{\Psi}^\eps_{\textrm{far}}}
\newcommand{\errthb}{\tilde{\Psi}^\eps_{\textrm{far},b}}

\newcommand{\eps}{\varepsilon}

\newcommand{\bigo}{\mathcal{O}}

\newcommand{\D}{\partial}

\newtheorem{thm}{Theorem}[section]
\newtheorem{rem}[thm]{Remark}

\nc{\pulse}{%
  \begin{picture}(60,50)(0,-20)
  \qbezier(0,0)(5,0)(10,15)
  \qbezier(10,15)(15,30)(20,30)
  \qbezier(20,30)(22,30)(25,18)
  \qbezier(25,18)(27,10)(30,10)
  \qbezier(30,10)(34,10)(38,15)
  \qbezier(38,15)(42,20)(45,20)
  \qbezier(45,20)(49,20)(53,10)
  \qbezier(53,10)(57,0)(60,0)
  \multiput(0,0)(0,-2){10}{\line(0,-1){1}}
  \multiput(60,0)(0,-2){10}{\line(0,-1){1}}
  \put(30,-20){\vector(-1,0){30}}
  \put(30,-20){\vector(1,0){30}}
  \put(0,-15){\makebox(60,10){{\tiny $T$}}}
  \end{picture}}